\documentclass[aps,pra,onecolumn,12pt,superscriptaddress,longbibliography,floatfix,nofootinbib]{revtex4-2}
\usepackage[T1]{fontenc}
\usepackage{graphicx}
\usepackage{xcolor}
\usepackage{amsmath}
\usepackage{amssymb}
\usepackage{braket}
\usepackage{booktabs}
\usepackage{amsthm}
\usepackage{algpseudocode}
\usepackage{silence}
\usepackage{hyperref}
\DeclareMathOperator{\Tr}{Tr}
\providecommand{\op}[2]{\ket{#1}\!\bra{#2}}
\newtheorem{theorem}{Theorem}
\newtheorem{lemma}{Lemma}
\newtheorem{corollary}{Corollary}
\newcounter{algorithm}
\renewcommand{\thealgorithm}{\arabic{algorithm}}

\graphicspath{{./}{figures/}}

\makeatletter
\def\@bibdataout@aps{%
 \immediate\write\@bibdataout{%
  @CONTROL{%
   apsrev42Control%
   \longbibliography@sw{%
    ,author="48",editor="1",pages="0",title="0",year="1"%
   }{%
    ,author="48",editor="1",pages="0",title="",year="1"%
   }%
  }%
 }%
 \if@filesw
  \immediate\write\@auxout{\string\citation{apsrev42Control}}%
 \fi
}%
\newif\ifappendixtoc@active
\newif\ifappendixtoc@printing
\newcommand{\appendixtableofcontents}{%
  \begingroup
    \let\addcontentsline\@gobblethree
    \section*{Appendix Contents}%
  \endgroup
  \begingroup
    \appendixtoc@activefalse
    \appendixtoc@printingfalse
    \let\appendixtoc@contentsline\contentsline
    \def\appendix{%
      \appendixtoc@activetrue
      \appendixtoc@printingtrue
    }%
    \def\tocdepth@munge{%
      \ifappendixtoc@active
        \appendixtoc@printingfalse
      \fi
    }%
    \def\tocdepth@restore{%
      \ifappendixtoc@active
        \appendixtoc@printingtrue
      \fi
    }%
    \def\contentsline##1##2##3##4{%
      \ifappendixtoc@printing
        \appendixtoc@contentsline{##1}{##2}{##3}{##4}%
      \fi
    }%
    \@starttoc{toc}%
  \endgroup
}
\makeatother

\begin{document}

\title{Zeno-Enhanced Probabilistic Error Cancellation with Quantum Error Detection Codes}
\author{Yi Yuan}
\affiliation{State Key Laboratory of Low Dimensional Quantum Physics, Department of Physics, Tsinghua University, Beijing, 100084, China}
\affiliation{Frontier Science Center for Quantum Information, Beijing 100184, China}

\author{Yuanchen Zhao}
\affiliation{State Key Laboratory of Low Dimensional Quantum Physics, Department of Physics, Tsinghua University, Beijing, 100084, China}
\affiliation{Frontier Science Center for Quantum Information, Beijing 100184, China}

\author{Dong E. Liu}
\email{Corresponding to: dongeliu@mail.tsinghua.edu.cn}
\affiliation{State Key Laboratory of Low Dimensional Quantum Physics, Department of Physics, Tsinghua University, Beijing, 100084, China}
\affiliation{Frontier Science Center for Quantum Information, Beijing 100184, China}
\affiliation{Beijing Academy of Quantum Information Sciences, Beijing 100193, China}
\affiliation{Hefei National Laboratory, Hefei 230088, China}

\begin{abstract}
Exact probabilistic error cancellation (PEC) is unbiased but suffers exponential sampling overhead set by noise-weighted circuit volume, whereas quantum error-detecting codes (QEDCs) remove many physical faults by stabilizer post-selection but leave an undetectable logical residue. We exploit this complementarity by using post-selection to map physical noise to a weaker accepted logical channel, and then applying PEC only to the residual channel. The resulting feedback-free QED+PEC scheme interleaves Clifford logical blocks, stabilizer measurements, post-selection, and probabilistic cancellation on accepted trajectories, without real-time decoding or active recovery. A key complication is that post-selection correlates accepted fault branches through stabilizer-commutation constraints, so the sparse Pauli-Lindblad factorization underlying bare PEC no longer applies directly. We therefore construct the inverse channel perturbatively: for fixed order $K$, only accepted fault branches up to order $K$ are retained, reducing enumeration from $2^m$ branches to $O(m^K)$ per block. The order-$K$ protocol cancels the normalized post-selected channel through degree $K$, leaving a per-block error $O(W^{K+1})$ under the stated assumptions, with an explicit end-to-end error bound. For logical GHZ-state preparation with the $[[n,n-2,2]]$ Iceberg code under circuit-level depolarizing noise and ideal stabilizer measurements, first-order QED+PEC reaches $n=200$ physical qubits and lowers sampling overhead by three to four orders of magnitude relative to standard PEC while maintaining $F\simeq0.956$. Syndrome-noise tests show that readout-only flips mainly increase post-selection cost, whereas noisy GHZ-assisted global stabilizer extraction can remove the advantage. This identifies a discrete-Zeno trade-off: cheap detection reshapes the effective channel PEC must invert, rather than simply adding overhead.
\end{abstract}
\maketitle
\section{Introduction}

Current quantum computing hardware sits between the noisy intermediate-scale quantum (NISQ) regime~\cite{5} and full fault-tolerant quantum computing (FTQC), a transitional resource regime often framed as early fault-tolerant quantum computing (EFTQC)~\cite{katabarwa2024early}. On the experimental side, this transition is already visible: utility-scale NISQ experiments have produced expectation-value estimates whose evaluation is beyond exact classical simulation~\cite{kim2023evidence}, while early logical processors built from superconducting~\cite{google2023suppressing,google2025quantum}, trapped-ion~\cite{iqbal2024non}, and reconfigurable neutral-atom~\cite{bluvstein2024logical,bluvstein2026fault} platforms have demonstrated logical qubits and small fault-tolerant primitives. Quantum error correction (QEC) provides an asymptotic route to scalability, and the threshold theorem~\cite{knill1997theory,aharonov1997fault,knill1998resilient,shor1996fault,preskill1998fault} guarantees that, once physical error rates fall below a hardware-specific threshold, arbitrary quantum computation can be performed reliably by simply scaling the code distance. However, the code length, syndrome-extraction depth, decoder latency, and physical-qubit overhead required to reach algorithmic-level fault tolerance remain substantial~\cite{10,google2025quantum,12}. The pressing near-term question is therefore how much useful computational power can be extracted before full fault tolerance becomes available.

Quantum error mitigation (QEM) was developed precisely to address this question on NISQ-era hardware, where active QEC is not yet feasible but a partial reduction of noise-induced bias is still desirable~\cite{QEM_review,aharonov2025importance}. A range of QEM schemes tailored to specific quantum algorithms has also been developed~\cite{gu2023noise,o2021error,kiss2025quantum}. QEM typically does not rely on complete real-time error-correction feedback, but it can be combined with coding, symmetry verification, and logic-level post-processing. Representative scalable or logic-level mitigation strategies include virtual distillation, Clifford data regression, scalable zero-noise-extrapolation workflows, and symmetry/subspace verification~\cite{huggins2021virtual,
czarnik2021error,kim2023scalable,27,29,30}. Probabilistic error cancellation (PEC) represents an inverse noise map as a quasiprobability mixture of realizable operations; exact PEC yields an unbiased estimator under its noise-model and implementation assumptions~\cite{13,14,QEM_review}. For unbiased mitigation schemes such as exact PEC, the price is statistical: the sampling overhead grows exponentially with the noise-weighted circuit volume, and this exponential blow-up is a fundamental limitation of unbiased QEM rather than a peculiarity of PEC~\cite{16,17,18,tsubouchi2023universal}. As soon as the circuit volume becomes moderate, QEM alone is therefore not enough to close the gap to FTQC. These lower bounds are important for interpreting the present work. We do not claim to evade the exponential sampling overhead required by unbiased QEM estimators~\cite{16,17,18,tsubouchi2023universal}. Instead, the role of QEDC is to modify the effective channel seen by PEC. In the first-order setting, bare PEC pays a quasiprobability factor controlled by a bare noise weight $W$, whereas QED post-selection replaces this by a residual accepted weight $\widetilde W$ in the PEC variance factor, at the price of an additional acceptance factor $p_{\mathrm{success}}^{-1}$. Thus the relevant question is not whether the exponential law disappears, but whether the exponent coefficient and the bias--cost operating point improve for the target task and hardware regime.

The fundamental limits of QEM can be traced to two structural ingredients of QEC that QEM lacks. First, QEC encodes the logical state in an entangled stabilizer subspace and introduces syndrome-resolved subspace filtering, so that physical noise becomes detectable through stabilizer syndromes. Second, QEC uses measurement-conditioned feedback to actively correct errors before they accumulate. Real-time decoding and feedback are demanding on near-term hardware, but the entangled-subspace structure can be incorporated alone, in a feedback-free manner, by using a quantum error-detecting code (QEDC): one encodes into a stabilizer subspace, measures the stabilizers, and post-selects on the trivial syndrome~\cite{19,20}. Distance-two QEDCs in particular detect any single qubit fault at modest hardware overhead. They cannot, however, eliminate residual logical noise from undetected multi-fault events or fault propagation, and their post-selection acceptance probability decays as the circuit volume grows.

This sets up the central question of the present work. Pure PEC has prohibitive sampling overhead; pure QEDC has irreducible residual bias and a shrinking acceptance probability. Can these two approaches be combined into an intermediate protocol whose sampling cost is lower than that of pure PEC, while its bias is lower than that of pure QEDC? Such a hybrid incurs higher sampling cost than QEDC alone (since PEC is still applied). For a finite-order implementation, it also retains bias relative to exact PEC, because only a truncated inverse of the post-selected logical channel is implemented. The relevant question is whether the joint protocol can outperform either component in regimes of practical interest. We refer to this construction as QED+PEC.

A key distinction in this work is between the abstract stabilizer projection and its physical syndrome-extraction circuit. Most of our theory and large-scale simulations first assume ideal stabilizer measurements. This isolates the logical mechanism of QED+PEC and corresponds to a cheap-detection regime. We later test two nonideal syndrome models and show that, for the Iceberg-code benchmark, extraction noise is the dominant practical limitation.

The main obstacle to constructing such a protocol is the scalability of the inverse channel. PEC requires inverting the noise channel acting on a logical block. For a generic noise model the number of independent process parameters grows exponentially with the number of logical qubits, so direct inversion via process tomography is classically intractable~\cite{23,24}. The sparse Pauli--Lindblad framework of Ref.~\cite{IBM2022} makes bare PEC scalable on NISQ devices by exploiting the fact that hardware noise is well captured by a small set of local Pauli generators, so that the inverse channel factorizes across these generators. This factorization is, however, broken by error detection: the post-selection acceptance condition couples contributions from different fault locations through stabilizer commutation, and naively expanding all sparse fault subsets within a block produces exponentially many branches, each with a different post-selected weight.

Our key technical insight is that the physically relevant regime is perturbative in the per-block noise weight. Let $m=O(cTn)$ denote the number of independent elementary Pauli fault mechanisms in a QED+PEC block, where $n$ is the number of physical qubits, $T$ is the number of logical layers between two detection rounds, and $c$ is the physical-layer compilation overhead. Rather than enumerating all $2^m$ fault subsets, we retain only branches with at most $K$ faults; the number of retained branches is $O(m^K)$, which is polynomial in $n$ for any fixed perturbative order $K$. In the Clifford logical setting studied in this paper, Pauli faults propagate as Pauli faults, and stabilizer commutation can be checked efficiently using the standard stabilizer-tableau formalism~\cite{19,26}. Thus, even though post-selection breaks the independent-generator factorization, the perturbative inverse can be compiled at any fixed $K$ under the stated assumptions; the full preprocessing includes the propagation and stabilizer-check costs detailed in Sec.~\ref{app:k_order_classical_complexity}.

In this work, we make the following contributions.

\begin{itemize}
    \item \textbf{Feedback-free QED+PEC protocol.}
    We introduce a protocol that alternates blocks of noisy logical computation with QEDC stabilizer measurements and post-selection, followed by PEC on the surviving logical channel. The protocol requires no real-time recovery operation: failed trajectories may be discarded either during execution or in classical post-processing.

    \item \textbf{Perturbative inverse channel of arbitrary order with polynomial preprocessing.}
    We construct the post-selected inverse channel order by order in the per-block total noise weight $W$. For any fixed perturbative order $K$, the number of retained fault branches scales as $O(m^K)$, and the classical preprocessing cost is polynomial in $n$, $T$, and circuit depth. The order-$K$ implementation cancels the aggregate contribution of accepted fault branches through degree $K$ in the normalized post-selected channel, leaving an $O(W^{K+1})$ channel error per QED+PEC block under the stated assumptions. We develop the construction for general $K$ and present the first-order implementation in full detail as the simplest nontrivial case.

    \item \textbf{Perturbative error upper bound and sampling-cost decomposition.}
    For the order-$K$ protocol, we prove a diamond-norm upper bound: under the assumptions of Theorem~\ref{thm:k_order_error_bound} and Corollary~\ref{cor:perturbative_scaling}, the fixed-order inverse cancels the normalized post-selected channel through degree $K$, leaving an $O(W^{K+1})$ per-block residual, with the end-to-end bound in Eq.~\eqref{eq:end_to_end_algorithmic_bound_main} and a linear-order accumulation estimate when $\sum_m\epsilon_{m,K}\ll1$. We also analyze the total sampling overhead as the product of a post-selection penalty and a PEC quasiprobability variance factor, showing explicitly how QEDC can reduce the PEC cost by shrinking the accepted residual noise weight $\widetilde W$.

    \item \textbf{Large-scale ideal-measurement Iceberg-code benchmark and syndrome-noise stress tests.}
    We showcase the first-order protocol on logical GHZ-state preparation using the $[[n,n-2,2]]$ Iceberg code. Because the standard GHZ preparation circuit has depth $O(n)$, the total circuit volume scales quadratically with the number of physical qubits, providing a stringent stress test of the perturbative construction at increasing scale. For circuit-level depolarizing noise on the logical-gate blocks with
    \[
        p_1 = 10^{-4}, \qquad p_2 = 10^{-3},
    \]
    and ideal stabilizer measurements, we simulate systems up to $n=200$ physical qubits, encoding $n-2$ logical qubits. Beyond a single operating point, we systematically vary the detection interval $T$, i.e., the number of logical layers between consecutive QEDC rounds, and chart its joint effect on the residual error and the sampling overhead. In this ideal-measurement regime, more frequent error detection \emph{simultaneously} improves the output fidelity and reduces the total sampling cost. At $n=200$ and detection interval $T=1$, QED+PEC reduces the sampling overhead by roughly three to four orders of magnitude relative to standard PEC while maintaining output fidelity around $F \simeq 0.956$. We then test two departures from the ideal-measurement assumption. A readout-only syndrome-flip model leaves the fidelity statistically unchanged but increases the observed post-selection cost. A circuit-level noisy GHZ-assisted extraction of the global Iceberg stabilizers, in contrast, can exceed the cost of pure PEC; a fixed-$n$ sweep shows that the cost is minimized at a finite detection interval rather than at the smallest tested $T$. The large ideal-measurement advantage is therefore to be read as a demonstration of the QED+PEC mechanism, not as a claim that naive noisy extraction of Iceberg checks is already advantageous on hardware.

    \item \textbf{Zeno-enhanced physical picture.}
    We introduce analytically solvable toy models that isolate the sampling-cost trade-off. Pure PEC prefers infrequent cancellation, because splitting a noise channel into many pieces compounds quasiprobability overhead. In contrast, in an ideal-projection model, QED+PEC benefits from frequent detection: although more frequent projections discard more shots, they reduce the residual logical noise faster than they increase the post-selection penalty. In the idealized model, the frequent-detection limit scales as
    \[
        \exp\!\left[\gamma T\left(1+\frac{4}{N}\right)\right],
    \]
    whereas optimal pure PEC scales as
    \[
        \exp(2\gamma T),
    \]
    giving an exponential separation for $N>4$. This behavior is a discrete analogue of the quantum Zeno effect and quantum
Zeno dynamics~\cite{36,37,facchi2008quantum}. This Zeno-like picture applies to the cheap-detection limit; noisy extraction adds an extra acceptance penalty that can shift the optimal detection interval to a finite value.
\end{itemize}

Our work is closely connected to several lines of prior research. The sparse Pauli--Lindblad PEC framework of van den Berg, Minev, Kandala, and Temme made unencoded PEC scalable under sparse-noise assumptions~\cite{IBM2022}. We extend this sparse-noise perspective to the post-selected logical subspace, where the stabilizer acceptance condition must be treated explicitly through a perturbative expansion that can be carried to any fixed order $K$. The Iceberg-code work of Self, Benedetti, and Amaro showed that high-rate distance-two QEDCs can protect expressive circuits on current hardware~\cite{20}. Our protocol complements this approach by canceling the undetected logical noise that remains after QEDC. Piveteau et al. and Suzuki et al. studied combinations of encoded computation and error mitigation in settings closer to error-corrected logical circuits~\cite{21,22}, while symmetry verification and quantum subspace expansion suppress errors through post-processing based on symmetries or subspace projections~\cite{27,IBM2022,29,30}. In contrast, our contribution is a blockwise, feedback-free QED+PEC construction with an explicit perturbative inverse channel and polynomial classical preprocessing at any fixed order $K$. Recent and concurrent works have also begun to explore direct combinations of QEDC and PEC~\cite{32}. Papi{\v c} et al.\ combined stabilizer-based error detection with PEC in fermionic encodings by post-selecting detectable errors and mitigating undetectable ones~\cite{papic2026snt}; our work provides an arbitrary-order and, under suitable assumptions, classically efficient framework for such QED+PEC protocols. Concretely, we analyze the post-selected sparse-fault expansion, derive explicit QED+PEC sampling rules and error bounds at arbitrary perturbative order, and identify the detection-frequency advantage through both large-$n$ Iceberg-code simulations and analytically solvable Zeno-type models. Finally, recent work on PEC scheduling without error detection shows that pure PEC generally favors fewer cancellation steps~\cite{31}. Our analysis shows that introducing error detection reverses this trade-off, making frequent detection beneficial.

\section{Background}

\subsection{Quantum error mitigation and probabilistic error cancellation}
Quantum error mitigation (QEM) provides a vital strategy for noisy quantum processors.
Rather than protecting the entire quantum state throughout the computation, QEM uses information about the noise process to classically postprocess measurement data and recover, as accurately as possible, the expectation value of a target observable~\cite{13,14,QEM_review}.
Among QEM protocols, probabilistic error cancellation (PEC) is one of the most direct: it represents the inverse of the noise channel as a quasiprobabilistic linear combination of physically implementable operations.
Suppose that the noisy implementation of an ideal operation is described by a channel $\mathcal{N}$.
Formally, if $\mathcal{N}$ is invertible on the relevant operator space, one may write an inverse map $\mathcal{N}^{-1}$.
This inverse is generally Hermiticity-preserving and trace-preserving, but it need not be completely positive, and hence cannot be implemented deterministically as a physical quantum channel.
PEC circumvents this obstruction by decomposing the inverse into a signed combination of implementable operations:
\begin{equation}
    \mathcal{N}^{-1}(\rho)
    =
    \sum_i q_i \mathcal{P}_i(\rho),
    \qquad
    \mathcal{P}_i(\rho)=P_i\rho P_i^\dagger ,
\end{equation}
where the quasiprobabilities $q_i$ may be positive or negative.
Defining
\begin{equation}
    \gamma = \sum_i |q_i|,
    \qquad
    p_i = \frac{|q_i|}{\gamma},
    \qquad
    s_i = \mathrm{sign}(q_i),
\end{equation}
we can rewrite the inverse channel as
\begin{equation}
    \mathcal{N}^{-1}(\rho)
    =
    \gamma
    \sum_i p_i s_i \mathcal{P}_i(\rho).
\end{equation}
Operationally, one samples $\mathcal{P}_i$ with probability $p_i$, runs the corresponding noisy circuit, and multiplies the measured estimator by the classical factor $\gamma s_i$.
The expectation value of this signed estimator reproduces the result of the ideal, noise-inverted evolution.
However, the price of this unbiased reconstruction is an increased variance: to estimate an observable to additive accuracy $\epsilon$, the required number of samples scales as
\begin{equation}
    N_{\rm samp}=O\!\left(\frac{\gamma^2}{\epsilon^2}\right),
\end{equation}
up to observable-dependent constants.

\subsection{Quantum error detection code}
While QEM addresses noise through classical postprocessing, an alternative hardware-level approach is necessary for fault tolerance. Quantum error correction (QEC) protects quantum information by extracting error syndromes and applying recovery operations, thereby preserving the encoded state throughout the computation~\cite{6,7,8,19}.
Quantum error detection (QEDC) takes a simpler route: instead of correcting the state after a nontrivial syndrome is observed, it post-selects on the trivial-syndrome subspace and discards all detected faulty runs.
This makes QEDC particularly attractive in near-term settings, where implementing reliable real-time correction may be more expensive than rejecting corrupted samples. Stabilizer-based error detection has also been demonstrated experimentally in small repetition-code devices~\cite{riste2015detecting}.

We briefly recall the stabilizer formulation.
Let $\mathcal{S}$ be an Abelian stabilizer group and let
\begin{equation}
    \Pi = \frac{1}{|\mathcal{S}|}\sum_{S\in\mathcal{S}} S
\end{equation}
be the projector onto the codespace.
For an input state $\rho$ supported on the codespace, consider a Pauli noise channel
\begin{equation}
    \mathcal{N}(\rho)=\sum_i w_i P_i \rho P_i^\dagger ,
    \qquad
    w_i\ge 0,\quad \sum_i w_i=1 ,
\end{equation}
where $P_i$ are Pauli operators.
After syndrome measurement and post-selection on the codespace, the survival probability is
\begin{equation}
    p_{\rm success}
    =
    \Tr\!\left[\Pi \mathcal{N}(\rho)\right],
\end{equation}
and the normalized post-selected state is
\begin{equation}
    \rho_{\rm QEDC}
    =
    \frac{\Pi \mathcal{N}(\rho)\Pi}
    {\Tr\!\left[\Pi \mathcal{N}(\rho)\right]} .
\end{equation}
Thus, QEDC suppresses noise by conditioning on successful syndrome outcomes, rather than by implementing an explicit recovery map.

The effectiveness of this procedure is controlled by the code distance $d$.
For a distance-$d$ code, QEC can correct arbitrary errors of weight up to $\lfloor (d-1)/2 \rfloor$, whereas QEDC can detect all errors of weight strictly smaller than $d$, provided that these errors take the state out of the codespace rather than acting as logical operators.
The price is statistical: post-selection reduces the number of accepted shots by a factor set by $p_{\rm suc}$, and therefore increases the sampling cost.
A prominent high-rate QEDC relevant to near-term demonstrations is the Iceberg code, designed to protect expressive circuits with modest qubit overhead~\cite{20}.
Its role is not to implement full fault-tolerant correction, but to remove a large class of low-weight faults through stabilizer checks and post-selection.

Consequently, PEC mitigates noise through a quasiprobability decomposition of the inverse noise channel, whereas QEDC suppresses errors by post-selecting on trivial-syndrome outcomes and discarding the remaining trajectories. Both schemes incur a sampling overhead as the price of noise suppression. PEC is in principle able to yield an unbiased estimator of the ideal expectation value but at a substantially larger sampling cost, whereas QEDC mitigates only those errors that trigger a nontrivial syndrome and is therefore blind to the undetectable component, yet incurs a significantly smaller overhead. A curious question therefore arises: does an intermediate regime exist that interpolates between the two, suppressing the residual bias more aggressively than QEDC alone while keeping the sampling overhead well below that of full PEC?

In what follows, we separate the ideal action of stabilizer post-selection from the physical cost of measuring the stabilizers. This distinction is important for codes with high-weight checks, where the extraction circuit itself can dominate the sampling cost.

\section{The Feedback-Free QED+PEC Protocol}
\label{sec:protocol}

We propose a hybrid approach that integrates quantum error detection codes (QEDC) with probabilistic error cancellation (PEC). The system is first encoded into a QEDC, and all subsequent computations are performed using logical gates within the code subspace. In our protocol, we restrict to Clifford logical gates, which ensures that propagated Pauli errors remain Pauli operators and can be efficiently tracked classically (see Sec.~\ref{app:k_order_classical_complexity}).

\subsection{Noise model}
\label{sec:noise_model}

We adopt a circuit-level independent sparse Pauli noise model, in which noise channels are applied locally depending on the specific operations performed on each qubit during a given physical gate layer. Denoting the ideal $l$th gate layer as $\mathcal{U}_l$, the noisy gate is
\begin{equation}
\tilde{\mathcal{U}}_l = \mathcal{U}_l \circ N_l\,,
\label{eq:noisy_gate}
\end{equation}
where the layer noise channel $N_l$ is the tensor product of local noise channels acting on individual qubits or pairs of qubits.

Specifically, if a qubit undergoes a single-qubit gate or remains idle during the $l$th layer, it experiences a single-qubit Pauli error channel:
\begin{equation}
\begin{split}
N_{\text{1Q}}(\rho) = &[(1-w_{1x})id+w_{1x}\mathcal{X})]\circ [(1-w_{1y})id+w_{1y}\mathcal{Y}]\circ \\
&[(1-w_{1z})id+w_{1z}\mathcal{Z}](\rho)
\end{split}
\end{equation}

If a pair of qubits undergoes a two-qubit gate, they jointly experience a two-qubit Pauli error channel:
\begin{equation}
N_{\text{2Q}}(\rho) = \prod_{P \in \{I, X, Y, Z\}^{\otimes 2} \setminus \{I \otimes I\}}[(1-w_{2P})id+w_{2P}\mathcal{P}](\rho)
\end{equation}

Consequently, the noise channel $N_l$ for the entire layer factorizes as a product of independent Pauli channels, one for each local noise source active in that layer. Let $\mathcal{K}_l$ denote the set of local Pauli error operators in the $l$th layer. The exact noise channel is
\begin{equation}
N_l(\rho) = \prod_{i \in \mathcal{K}_l} [(1 - w_i)id + w_i\, \mathcal{P}_i ](\rho)
\label{eq:noise_channel}
\end{equation}
where $w_i$ is the error probability associated with the Pauli operator $P_i$, and the product is understood as a composition of superoperators. This Pauli-channel assumption is also compatible with noise-tailoring techniques such as randomized compiling, which tailor coherent errors into effective stochastic Pauli noise~\cite{wallman2016noise}.

In particular, when $w_{1x}=w_{1y}=w_{1z}$, $N_{\text{1Q}}$ reduces to single-qubit depolarizing noise; choosing all nonidentity two-qubit weights equal gives the analogous two-qubit depolarizing channel. We use this depolarizing specialization in the numerical demonstration below.

Unless explicitly stated otherwise, this noise model is applied only to the logical-gate blocks; the stabilizer measurements used for post-selection are treated as ideal projective measurements. We separately test readout-only syndrome flips and a circuit-level noisy GHZ-assisted extraction circuit in Sec.~\ref{sec:syndrome_noise_checks}.

\subsection{Protocol overview}
\label{sec:protocol_overview}

After every $T$ layers of noisy logical gates, an error detection (QED) step is performed by measuring all stabilizers of the QEDC. In the main construction this measurement is modeled as an ideal stabilizer projection; noisy readout and noisy extraction circuits are considered separately in Sec.~\ref{sec:syndrome_noise_checks}. Quantum trajectories that yield nontrivial syndromes---indicating that the state has leaked out of the code subspace---are discarded via post-selection. However, QED alone does not eliminate all noise: certain error processes preserve the code subspace and manifest as residual logical noise after post-selection. To cancel this residual noise, we apply a PEC-based error mitigation step on the logical subspace. This QED+PEC cycle is then repeated for the remainder of the circuit (see Fig.~\ref{fig:QED_PEC_circuit}).

\begin{figure}[t]
    \centering
    \includegraphics[width=\columnwidth]{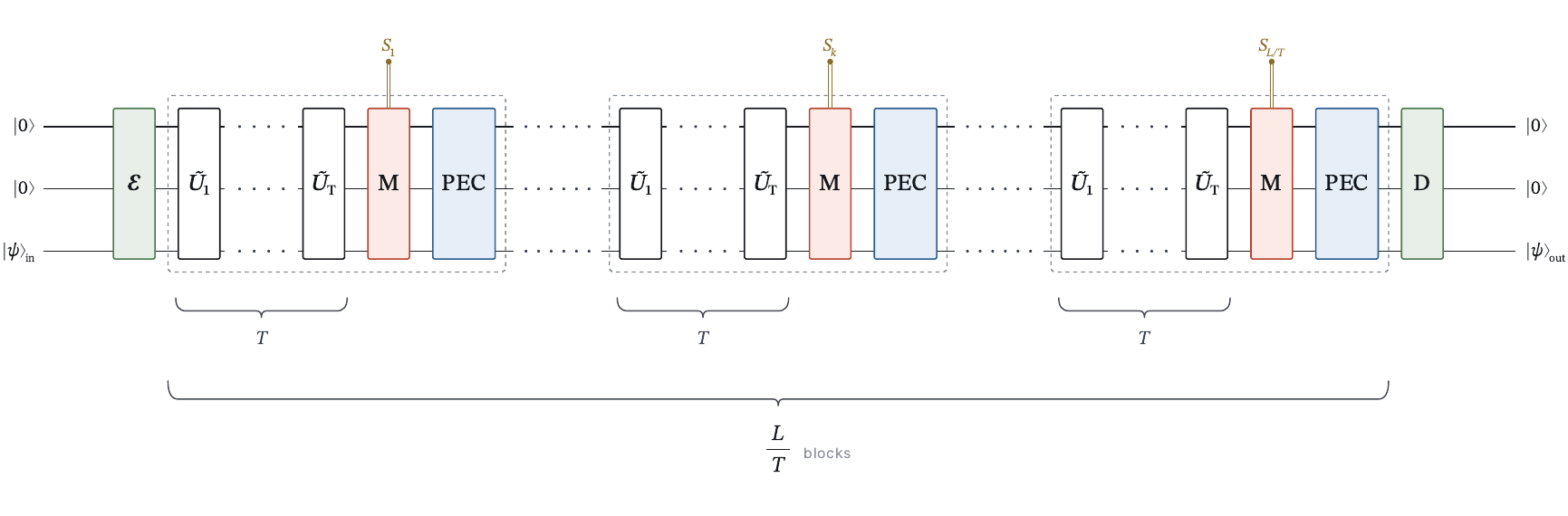}
    \caption{Quantum circuit implementing the QED+PEC protocol. After every $T$ layers of logical gates, stabilizer measurements perform error detection. Trajectories passing QED then undergo PEC to cancel residual logical noise.}
    \label{fig:QED_PEC_circuit}
\end{figure}

This approach combines the strengths of quantum error correction and quantum error mitigation while circumventing key limitations of each. It exploits the entanglement structure of the QEDC subspace to suppress local noise, but does not require active error correction based on mid-circuit measurement feedback---a significant hardware challenge. Instead, the residual accepted logical noise is mitigated through PEC sampling. Meanwhile, it inherits the resource trade-off of QEM: exchanging temporal resources (sampling overhead) for spatial resources (physical qubits). The complete procedure is summarized in Algorithm~\ref{alg:edpec}. We note that the protocol takes the noise weights $\{w_i\}$ as known inputs; in practice, these are obtained from a prior noise characterization step, such as the sparse Pauli--Lindblad noise learning procedure demonstrated in Ref.~\cite{IBM2022}. The sampling-cost figures reported below count only the online repetition overhead from PEC quasiprobabilities and QED post-selection; the one-time characterization and PEC-table compilation costs are not included.

\begingroup
\small
\refstepcounter{algorithm}\label{alg:edpec}
\noindent
\textbf{Algorithm \thealgorithm.} $K$th-Order Feedback-Free QED+PEC Protocol
\medskip
\begin{algorithmic}[1]
\Require $n$-qubit QEDC with stabilizers $\{S_1,\ldots,S_s\}$ and projector $\Pi$; $L$-layer logical Clifford circuit $\mathcal U_1,\ldots,\mathcal U_L$; independent Bernoulli Pauli fault locations with weights $\{w_i\}$; detection interval $T$; perturbative order $K$; observable $O$; number of attempted samples $N_{\mathrm{shots}}$
\Ensure estimator $\hat{o}$ of $\langle O\rangle_{\mathrm{ideal}}$
\State $M \gets L/T$ \Comment{number of QED+PEC cycles; assume $T$ divides $L$, otherwise use a shorter final block}

\Statex \textit{// --- Phase 1: Classical pre-computation of $K$th-order PEC tables ---}
\For{$m=1$ \textbf{to} $M$}
    \State $\mathcal F_m \gets$ all independent Pauli fault locations in block $m$
    \For{each elementary fault $i\in\mathcal F_m$ occurring after physical layer $\ell_i$}
        \State Propagate to the end of the block:
        \Statex \hspace{1.5em}$\tilde P_i^{(m)} \gets U_{mcT}\cdots U_{\ell_i}\,P_i\,U_{\ell_i}^{\dagger}\cdots U_{mcT}^{\dagger}$
    \EndFor

    \State $\bar N_{K,m}\gets 0$; \quad $P_{K,\mathrm{success}}^{(m)}\gets 0$
    \For{each subset $I\subseteq\mathcal F_m$ with $|I|\le K$}
        \State $w_I\gets\prod_{i\in I}w_i$; \quad $\tilde P_I^{(m)}\gets\prod_{i\in I}\tilde P_i^{(m)}$ \Comment{$w_\varnothing=1$, $\tilde P_\varnothing=I$}
        \State $a_{I,m}^{(K)}\gets w_I\displaystyle\sum_{\ell=0}^{K-|I|}(-1)^\ell
        \sum_{\substack{J\subseteq\mathcal F_m\setminus I\\ |J|=\ell}}\prod_{j\in J}w_j$
        \If{$[\tilde P_I^{(m)},S_a]=0$ for all $a=1,\ldots,s$}
            \State $\mathcal P_I^{(m)}(\cdot)\gets\tilde P_I^{(m)}(\cdot)\tilde P_I^{(m)\dagger}$
            \State $\bar N_{K,m}\gets\bar N_{K,m}+a_{I,m}^{(K)}\mathcal P_I^{(m)}$
            \State $P_{K,\mathrm{success}}^{(m)}\gets P_{K,\mathrm{success}}^{(m)}+a_{I,m}^{(K)}$
        \EndIf
    \EndFor

    \State $R_{K,m}\gets \operatorname{Trunc}_{\le K}\!\left[\bar N_{K,m}/P_{K,\mathrm{success}}^{(m)}\right]-\operatorname{id}$
    \State $\displaystyle \sum_P c_{P,m}^{(K)}\mathcal P \gets
    \operatorname{Collect}_{P}\operatorname{Trunc}_{\le K}\!\left[\sum_{r=0}^{K}(-1)^r R_{K,m}^{\circ r}\right]$
    \State $\gamma_m^{(K)}\gets\sum_{P}|c_{P,m}^{(K)}|$
    \State $q_{P,m}^{(K)}\gets |c_{P,m}^{(K)}|/\gamma_m^{(K)}$ for each $P$ with $c_{P,m}^{(K)}\ne0$
    \State Store $\mathsf T_m^{(K)}\gets\{(P,q_{P,m}^{(K)},\operatorname{sgn}(c_{P,m}^{(K)})):\ c_{P,m}^{(K)}\ne0\}$
\EndFor

\Statex \textit{// --- Phase 2: Quantum execution, post-selection, and $K$th-order PEC sampling ---}
\State $\mathcal A\gets\varnothing$ \Comment{accepted trajectories}
\For{$k=1$ \textbf{to} $N_{\mathrm{shots}}$}
    \State Encode into logical initial state $\rho=\Pi\rho\Pi$
    \State $\textit{sign}_k\gets +1$; \quad $\Gamma_k\gets 1$; \quad $\textit{accepted}\gets\mathrm{true}$
    \For{$m=1$ \textbf{to} $M$}
        \State Apply $T$ layers of noisy logical gates
        \Statex \hspace{1.5em}$\tilde{\mathcal U}_{(m-1)T+1},\ldots,\tilde{\mathcal U}_{mT}$
        \State Measure all stabilizers $S_1,\ldots,S_s$ \Comment{QEDC post-selection}
        \If{any syndrome $\ne +1$}
            \State $\textit{accepted}\gets\mathrm{false}$; \quad \textbf{break}
        \EndIf
        \State Sample $(P,q,s)$ from $\mathsf T_m^{(K)}$ with probability $q=q_{P,m}^{(K)}$
        \State Apply the Pauli channel $\mathcal P(\cdot)=P(\cdot)P^\dagger$
        \State $\textit{sign}_k\gets\textit{sign}_k\times s$; \quad $\Gamma_k\gets\Gamma_k\times\gamma_m^{(K)}$
    \EndFor
    \If{$\textit{accepted}$}
        \State Measure observable $O$ and obtain outcome $x_k$
        \State Record $o_k\gets \Gamma_k\,\textit{sign}_k\,x_k$; \quad $\mathcal A\gets\mathcal A\cup\{k\}$
    \EndIf
\EndFor
\State \Return $\displaystyle \hat{o}=\frac{1}{|\mathcal A|}\sum_{k\in\mathcal A}o_k$
\end{algorithmic}
\endgroup

\subsubsection{Error detection step}
\label{sec:ed_step}

Let $\Pi$ denote the projector onto the code subspace of the QEDC. The logical initial state satisfies $\rho = \Pi \rho \Pi$, and the ideal logical gates preserve the subspace: $\mathcal{U}_T \circ \cdots \circ \mathcal{U}_1(\rho) \in \Pi$. In the presence of noise, the state may partially leak out of the subspace. After $T$ layers of noisy gates $\tilde{\mathcal{U}}_l = \mathcal{U}_l \circ N_l$, the QED step projects back onto the code subspace:
\begin{equation}
\begin{aligned}
&\Pi \circ \tilde{\mathcal{U}}_T \circ \cdots
    \circ \tilde{\mathcal{U}}_1(\rho)
    = \Pi \circ N \circ \mathcal{U}(\rho)\\
&\quad
    = p_{\mathrm{success}}\,
    N_{\mathrm{reduced}} \circ \Pi \circ \mathcal{U}(\rho)\\
&\quad
    = p_{\mathrm{success}}\,
    N_{\mathrm{reduced}} \circ \mathcal{U}(\rho)
\end{aligned}
\label{eq:ed_projection}
\end{equation}
where $p_{\mathrm{success}} \leq 1$ is the probability of passing QED. $\mathcal{U} = \mathcal{U}_T \circ \cdots \circ \mathcal{U}_1$ is the ideal unitary, and $N_{\mathrm{reduced}}$ is the residual noise channel that preserves the code subspace.

We remark that the post-selection in the error detection stage of Algorithm~\ref{alg:edpec} need not be implemented as a real-time decision during circuit execution. An equivalent strategy is to run all quantum trajectories to completion---including those that would fail QED---and filter in classical post-processing by retaining only the trajectories whose syndrome measurements all returned $+1$. The two approaches yield identical expectation values; the choice between them is dictated by experimental convenience and hardware constraints.

At the formal level, a noisy syndrome-extraction circuit can be absorbed into the same perturbative framework by adding its fault locations to the block and by replacing the ideal codespace projector with an accepted measurement instrument. At fixed perturbative order $K$, only fault subsets of size at most $K$ are enumerated, so this extension does not increase the preprocessing cost exponentially. It can, however, substantially reduce the observed acceptance probability and increase the total sampling cost. For high-weight stabilizers such as the Iceberg checks $X^{\otimes n}$ and $Z^{\otimes n}$, this extraction-induced cost can dominate the PEC saving, as shown in Sec.~\ref{sec:syndrome_noise_checks}.

In a readout-only syndrome-flip model, a false rejection can occur at first order in the readout error probability and only reduces the number of accepted shots. False acceptance of a detectable single physical fault, in contrast, requires both the physical fault and a readout flip that masks its syndrome, and is therefore second order in the joint fault expansion. This explains why readout-only errors mainly affect the observed post-selection denominator in the first-order simulations below.

\subsubsection{PEC step: \texorpdfstring{$K$th-order}{Kth-order} inverse channel}
\label{sec:high_order}
To cancel the residual noise $N_{\mathrm{reduced}}$, we construct an approximate inverse channel using PEC.

Consider one QED+PEC block whose noise-free logical action is
\begin{equation}
    \mathcal U := \mathcal U_T\circ\cdots\circ \mathcal U_1 ,
    \qquad
    \rho' := \mathcal U(\rho),
    \qquad
    \rho=\Pi\rho\Pi .
\end{equation}
Under the assumption of independent sparse Pauli fault model, the index set $\mathcal K_T$ enumerates independent physical Pauli fault locations during the block, each occurring with probability $w_i\in[0,1)$. After propagating every physical Pauli fault to the end of the block, denote the propagated Pauli by $\tilde P_i$ and define
\begin{equation}
    \mathcal P_i(\cdot):=\tilde P_i(\cdot)\tilde P_i^\dagger,
    \qquad
    W:=\sum_{i\in\mathcal K_T} w_i .
    \label{eq:total_noise}
\end{equation}
For a subset $I\subseteq\mathcal K_T$, define
\begin{equation}
\begin{split}
    w_I:=\prod_{i\in I}w_i,
    \qquad
    &\tilde P_I:=\prod_{i\in I}\tilde P_i,
    \qquad
    \mathcal P_I(\cdot):=\tilde P_I(\cdot)\tilde P_I^\dagger\\
    &\pi_I:=w_I\prod_{j\in\mathcal K_T\setminus I}(1-w_j)
\end{split}
\end{equation}
with $w_\varnothing=1$ and $\mathcal P_\varnothing=\operatorname{id}$. The propagated Pauli noise channel acting on $\rho'$ then factorizes as
\begin{equation}
    N
    =
    \prod_{i\in\mathcal K_T}\!\left[(1-w_i)\operatorname{id}+w_i\mathcal P_i\right]
    =
    \sum_{I\subseteq\mathcal K_T}
    \pi_I\mathcal P_I .
\end{equation}
Let $\operatorname{Trunc}_{\le K}$ denote truncation to total degree at most $K$ in the weights $\{w_i\}$. For $|I|\le K$, define the $K$th-order branch coefficient
\begin{equation}
\begin{split}
    a_I^{(K)}
    &:=
    \operatorname{Trunc}_{\le K}
    \left[
    w_I\prod_{j\in\mathcal K_T\setminus I}(1-w_j)
    \right]\\
    &=
    w_I
    \sum_{\ell=0}^{K-|I|}
    (-1)^\ell
    \sum_{\substack{
        J\subseteq\mathcal K_T\setminus I\\
        |J|=\ell
    }}
    w_J
\end{split}
\end{equation}
Then the $K$th-order noise channel is
\begin{equation}
    N_K
    =
    \sum_{\substack{
        I\subseteq\mathcal K_T\\
        |I|\le K
    }}
    a_I^{(K)}\mathcal P_I
    =
    N+O(W^{K+1}) .
\end{equation}

Define the QEDC-pass fault set, which is independent of the truncation order, and its $K$th-order subset
\begin{equation}
\begin{split}
    \mathcal K_{\mathrm{pass}}
    &:=
    \left\{
    I\subseteq\mathcal K_T:
    [\tilde P_I,S_a]=0,\ \forall a
    \right\}\\
    &\mathcal K_{\mathrm{success}}^{(\le K)}
    :=
    \mathcal K_{\mathrm{pass}}
    \cap
    \{I:|I|\le K\}
\end{split}
\end{equation}
Whether a fault product $\tilde P_I$ survives the post-selection is purely a commutation condition with the stabilizer generators and does not depend on $K$. Concretely, for $\rho'=\Pi\rho'\Pi$,
\begin{equation}
    \Pi\mathcal P_I(\rho')\Pi
    =
    \begin{cases}
    \mathcal P_I(\rho'), & I\in \mathcal K_{\mathrm{pass}},\\
    0, & I\notin \mathcal K_{\mathrm{pass}} .
    \end{cases}
\end{equation}
The exact all-order post-selected residual channel is
\begin{equation}
\begin{split}
    N_{\mathrm{reduce}}^{\mathrm{exact}}
    &:=
    \frac{1}{p_{\mathrm{success}}}
    \sum_{I\in\mathcal K_{\mathrm{pass}}}
    \pi_I\mathcal P_I\\
    p_{\mathrm{success}}
    &:=
    \operatorname{Tr}\Pi N(\rho')\Pi
    =
    \sum_{I\in\mathcal K_{\mathrm{pass}}}
    \pi_I
\end{split}
    \label{eq:exact_hardware_reduced_main}
\end{equation}
This is the channel that must be mitigated. It contains accepted branches of
every fault order, whereas the $K$th-order truncation below restricts the
post-selected sum to $I\in\mathcal K_{\mathrm{success}}^{(\le K)}$.

The unnormalized post-selected $K$th-order channel is therefore
\begin{equation}
    \bar N_K(\rho')
    :=
    \Pi N_K(\rho')\Pi
    =
    \sum_{I\in\mathcal K_{\mathrm{success}}^{(\le K)}}
    a_I^{(K)}\mathcal P_I(\rho') .
\end{equation}
The $K$th-order approximation to the survival probability is
\begin{equation}
    P_{K,\mathrm{success}}
    :=
    \operatorname{Tr}\bar N_K(\rho')
    =
    \sum_{I\in\mathcal K_{\mathrm{success}}^{(\le K)}}
    a_I^{(K)} .
\end{equation}
The exact post-selection rate $p_{\mathrm{success}}$ is approximated by $P_{K,\mathrm{success}}$ with error $p_{\mathrm{success}}-P_{K,\mathrm{success}}=O(W^{K+1})$.
The normalized surviving channel is
\begin{equation}
    N_{K,\mathrm{reduce}}
    :=
    \frac{\bar N_K}{P_{K,\mathrm{success}}}
    =
    \frac{
    \sum_{I\in\mathcal K_{\mathrm{success}}^{(\le K)}}
    a_I^{(K)}\mathcal P_I
    }{
    P_{K,\mathrm{success}}
    } .
\end{equation}
Equivalently,
\begin{equation}
    N_{K,\mathrm{reduce}}
    =
    \operatorname{id}
    +
    R_K
    +
    O(W^{K+1}),
\end{equation}
where
\begin{equation}
    R_K
    :=
    \operatorname{Trunc}_{\le K}
    \left[
    \frac{
    \sum_{I\in\mathcal K_{\mathrm{success}}^{(\le K)}}
    a_I^{(K)}
    \mathcal P_I
    }{
    P_{K,\mathrm{success}}
    }
    \right]
    -\operatorname{id}.
\end{equation}
Since $R_K=O(W)$, the $K$th-order inverse channel is given by the truncated geometric inverse
\begin{equation}
    N_{K,\mathrm{reduce}}^{-1}
    :=
    \operatorname{Trunc}_{\le K}
    \left[
    \sum_{r=0}^{K}
    (-1)^r R_K^{\circ r}
    \right],
\end{equation}
where
\begin{equation}
    R_K^{\circ 0}:=\operatorname{id},
    \qquad
    R_K^{\circ r}:=
    \underbrace{R_K\circ R_K\circ\cdots\circ R_K}_{r\ \mathrm{times}}
    \quad (r\ge 1).
\end{equation}
Therefore,
\begin{equation}
    N_{K,\mathrm{reduce}}^{-1}
    \circ
    N_{K,\mathrm{reduce}}
    =
    \operatorname{id}
    +
    O(W^{K+1}).
\end{equation}

Expanding $N_{K,\mathrm{reduce}}^{-1}$ and collecting identical Pauli channels, we write
\begin{equation}
    N_{K,\mathrm{reduce}}^{-1}
    =
    \sum_{P\in\mathsf P_n}
    c_P^{(K)}\mathcal P
    +
    O(W^{K+1}),
\end{equation}
where $\mathsf P_n$ is the $n$-qubit Pauli group modulo phases, $\mathcal P(\cdot)=P(\cdot)P^\dagger$, and
\begin{equation}
    c_P^{(K)}
    :=
    [\mathcal P]\,
    \operatorname{Trunc}_{\le K}
    \left[
    \sum_{r=0}^{K}
    (-1)^r R_K^{\circ r}
    \right].
\end{equation}
Here $[\mathcal P](\cdot)$ denotes the coefficient of the Pauli channel $\mathcal P$ after collecting identical Pauli channels.

The $K$th-order PEC norm is
\begin{equation}
    \gamma_K
    :=
    \sum_{P\in\mathsf P_n}
    \left|c_P^{(K)}\right|.
\end{equation}
Hence the PEC form of the inverse channel is
\begin{equation}
    N_{K,\mathrm{reduce}}^{-1}
    =
    \gamma_K
    \sum_{P\in\mathsf P_n}
    \operatorname{sgn}\!\left(c_P^{(K)}\right)
    \frac{\left|c_P^{(K)}\right|}{\gamma_K}
    \mathcal P
    +
    O(W^{K+1}).
\end{equation}
Thus, conditioned on passing QED, the $K$th-order PEC step samples the Pauli channel $\mathcal P$ with probability
\begin{equation}
    q_P^{(K)}
    =
    \frac{\left|c_P^{(K)}\right|}{\gamma_K},
\end{equation}
applies $\mathcal P$, and multiplies the measurement outcome by
\begin{equation}
    \gamma_K\,\operatorname{sgn}\!\left(c_P^{(K)}\right).
\end{equation}
The one-block sampling overhead can therefore be estimated, to $K$th order, as
\begin{equation}
    C_K
    =
    \frac{\gamma_K^2}{P_{K,\mathrm{success}}}.
\end{equation}

\subsection{Approximation validity and error bound}
\label{sec:validity}

We now state the approximation guarantee for the order-$K$ QED+PEC
protocol in the form implemented by Algorithm~\ref{alg:edpec}.
The key distinction is the following. After a successful QED round, the
hardware produces the all-order post-selected channel determined by the exact
Bernoulli probabilities.  The PEC table, however, is built from the
coefficients $a_{I,m}^{(K)}$ obtained by keeping terms only through order
$K$, and from the corresponding truncated formal Neumann inverse.  The theorem below
directly bounds the discrepancy between these two objects.

For the error analysis below, we use the block-indexed version of the notation
from Sec.~\ref{sec:high_order}: all quantities for block $m$ carry a subscript
$m$. Thus $N_{m,\mathrm{reduce}}^{\mathrm{exact}}$ denotes the exact channel in
Eq.~\eqref{eq:exact_hardware_reduced_main} for block $m$, while the
order-$K$ reduced channel used for PEC compilation and the sampled
quasiprobability inverse are
\begin{equation}
\begin{aligned}
    \hat{N}_{m,\mathrm{reduce}}
    &:=
    \operatorname{id}+R_{K,m},\\
    \widehat N_{m,K}^{-1}
    &:=
    \sum_P c_{P,m}^{(K)}\mathcal P,\\
    \gamma_m^{(K)}
    &:=
    \sum_P |c_{P,m}^{(K)}|.
\end{aligned}
\end{equation}
The protocol implements $\widehat N_{m,K}^{-1}$, rather than the exact inverse
of $N_{m,\mathrm{reduce}}^{\mathrm{exact}}$.

\begin{theorem}[Perturbative error upper bound for the order-$K$ QED+PEC protocol]
\label{thm:k_order_error_bound}
Assume independent Bernoulli Pauli faults, Clifford logical blocks, ideal
stabilizer measurements, and a stabilizer QED code.  For each block $m$, let
$N_{m,\mathrm{reduce}}^{\mathrm{exact}}$ be the exact all-order post-selected
hardware residual channel defined as the block-indexed version of
Eq.~\eqref{eq:exact_hardware_reduced_main}, let
$\hat{N}_{m,\mathrm{reduce}}$ and $\widehat N_{m,K}^{-1}$ be the corresponding
order-$K$ objects from Sec.~\ref{sec:high_order},
and define the two block certificates
\begin{align}
    \eta_{m,K}
    &:=
    \left\|
        N_{m,\mathrm{reduce}}^{\mathrm{exact}}-\hat{N}_{m,\mathrm{reduce}}
    \right\|_\diamond,
    \label{eq:eta_certificate_main}\\
    \zeta_{m,K}
    &:=
    \left\|
        \widehat N_{m,K}^{-1}\circ\hat{N}_{m,\mathrm{reduce}}
        -\operatorname{id}
    \right\|_\diamond.
    \label{eq:zeta_certificate_main}
\end{align}
Equivalently, one may replace $\eta_{m,K}$ and $\zeta_{m,K}$ by any certified
upper bounds on the two norms above.  Then one QED+PEC block satisfies
\begin{equation}
\begin{split}
    \left\|
        \widehat N_{m,K}^{-1}\circ N_{m,\mathrm{reduce}}^{\mathrm{exact}}
        -\operatorname{id}
    \right\|_\diamond
    \le
    \epsilon_{m,K}\\
    \epsilon_{m,K}
    :=
    \zeta_{m,K}+\gamma_m^{(K)}\eta_{m,K}.
\end{split}
\label{eq:one_block_algorithmic_bound_main}
\end{equation}
For an $M$-block circuit, conditioned on all QED rounds accepting, the channel
implemented by the QED+PEC estimator satisfies
\begin{equation}
\begin{split}
    \left\|
        \mathcal C_{\mathrm{QED+PEC}}^{(K)}
        -\mathcal C_{\mathrm{ideal}}
    \right\|_\diamond
    \le
    \prod_{m=1}^{M}(1+\epsilon_{m,K})-1\\
    \le\exp\!\left(\sum_{m=1}^{M}\epsilon_{m,K}\right)-1.
    \end{split}
    \label{eq:end_to_end_algorithmic_bound_main}
\end{equation}
Consequently, for any observable $O$ and input state in the code space,
\begin{equation}
\begin{split}
    \left|
        \mathbb E[\hat o\mid\mathrm{all\ QED\ rounds\ accepted}]
        -\operatorname{Tr}\!\big[O\,\mathcal C_{\mathrm{ideal}}(\rho)\big]
    \right|\\
    \le
    \|O\|_\infty
    \left[
        \prod_{m=1}^{M}(1+\epsilon_{m,K})-1
    \right].
    \label{eq:observable_bias_algorithmic_bound_main}
\end{split}
\end{equation}
\end{theorem}

\medskip
\noindent\textbf{Proof sketch.}
For one block, we separate the error of inverting the
order-$K$ channel from the error caused by replacing the exact hardware
channel by that order-$K$ channel.  The exact identity is
\begin{equation}
\begin{split}
    \widehat N_{m,K}^{-1}\circ N_{m,\mathrm{reduce}}^{\mathrm{exact}}
    -\operatorname{id}
    =
    \big(\widehat N_{m,K}^{-1}\circ\hat{N}_{m,\mathrm{reduce}}
        -\operatorname{id}\big)\\
    +
    \widehat N_{m,K}^{-1}\circ
    \big(N_{m,\mathrm{reduce}}^{\mathrm{exact}}
        -\hat{N}_{m,\mathrm{reduce}}\big).
\end{split}
\end{equation}
Taking diamond norms and using the triangle inequality gives
\begin{equation}
\begin{split}
    \left\|
        \widehat N_{m,K}^{-1}\circ  N_{m,\mathrm{reduce}}^{\mathrm{exact}}
        -\operatorname{id}
    \right\|_\diamond
    \le
    \zeta_{m,K}
    +\\
    \left\|
        \widehat N_{m,K}^{-1}\circ
        \big(N_{m,\mathrm{reduce}}^{\mathrm{exact}}
            -\hat{N}_{m,\mathrm{reduce}}\big)
    \right\|_\diamond
\end{split}
\end{equation}
For the second term,
\begin{equation}
    \left\|
        \widehat N_{m,K}^{-1}\circ
        \big(N_{m,\mathrm{reduce}}^{\mathrm{exact}}
            -\hat{N}_{m,\mathrm{reduce}}\big)
    \right\|_\diamond
    \le
    \|\widehat N_{m,K}^{-1}\|_\diamond\,\eta_{m,K}.
\end{equation}
Expanding the PEC inverse as
$\widehat N_{m,K}^{-1}=\sum_P c_{P,m}^{(K)}\mathcal P$ and using
$\|\mathcal P\|_\diamond=1$ for each unitary Pauli channel gives
\begin{equation}
    \|\widehat N_{m,K}^{-1}\|_\diamond
    \le
    \sum_P |c_{P,m}^{(K)}|
    =
    \gamma_m^{(K)}.
\end{equation}
Thus the one-block error is at most
$\zeta_{m,K}+\gamma_m^{(K)}\eta_{m,K}$.

For an $M$-block circuit, write
$\widehat N_{m,K}^{-1}\circ N_{m,\mathrm{reduce}}^{\mathrm{exact}}
=\operatorname{id}+\Delta_{m,K}$ with
$\|\Delta_{m,K}\|_\diamond\le\epsilon_{m,K}$.  After conjugating these
errors by the intervening ideal Clifford blocks, write the resulting
errors as $\widetilde\Delta_{m,K}$.  Their diamond norms are unchanged.
Hence the error relative to the ideal circuit is bounded by
\begin{equation}
\begin{split}
    \left\|
        \prod_{m=1}^{M}
        \big(\operatorname{id}+\widetilde\Delta_{m,K}\big)
        -\operatorname{id}
    \right\|_\diamond
    &\le
    \sum_{\emptyset\ne S\subseteq\{1,\ldots,M\}}
    \prod_{m\in S}\epsilon_{m,K}\\
    &=
    \prod_{m=1}^{M}(1+\epsilon_{m,K})-1
\end{split}
\end{equation}
The exponential bound follows from $1+x\le e^x$.  The observable bound
follows from
$|\operatorname{Tr}[O(\Phi-\Psi)(\rho)]|
\le \|O\|_\infty\|\Phi-\Psi\|_\diamond$.

\begin{corollary}[Perturbative scaling]
\label{cor:perturbative_scaling}
Under the assumptions of Theorem~\ref{thm:k_order_error_bound}, write
$W_m:=\sum_{i\in\mathcal K_m}w_i$ and assume the success probability
$p_m$ stays bounded away from zero as $W_m\to0$.  Then
\begin{equation}
\begin{split}
    \eta_{m,K}=O(&W_m^{K+1}),
    \qquad
    \zeta_{m,K}=O(W_m^{K+1})\\
    &\gamma_m^{(K)}=1+O(W_m)
\end{split}
\label{eq:perturbative_three_estimates}
\end{equation}
so the one-block certificate satisfies $\epsilon_{m,K}=O(W_m^{K+1})$,
and whenever $\sum_m\epsilon_{m,K}\ll1$,
\begin{equation}
    \left\|
        \mathcal C_{\mathrm{QED+PEC}}^{(K)}
        -\mathcal C_{\mathrm{ideal}}
    \right\|_\diamond
    =O\!\left(\sum_{m=1}^{M}W_m^{K+1}\right).
    \label{eq:end_to_end_perturbative_scaling}
\end{equation}
The last estimate is obtained by substituting
$\epsilon_{m,K}=O(W_m^{K+1})$ into
Eq.~\eqref{eq:end_to_end_algorithmic_bound_main} and using
$\prod_m(1+\epsilon_{m,K})-1
=\sum_m\epsilon_{m,K}+O((\sum_m\epsilon_{m,K})^2)$.
\end{corollary}

A simple intuitive criterion for choosing the truncation order is to take $K\le d-1$, where $d$ is the distance of the QEDC. QED eliminates fault products whose effective weight is lower than $d$. Thus, for single-qubit-level faults, QED alone suppresses the error rate up to $O(p^d)$. In a circuit-level noise model, however, fault propagation can generate high-weight fault operators with relatively low order in the physical noise strength. For PEC orders $K<d$, many retained fault products are still below the code distance and can therefore be filtered by QED, which reduces the PEC sampling overhead. In contrast, for $K\ge d$, the retained terms of order $O(W^K)$ contain a larger fraction of high-weight fault operators with weight at least $d$. These terms are not preferentially removed by a distance-$d$ detector, so QED provides less advantage in reducing the sample cost at these orders. Therefore, in the numerical section below, we focus on first-order truncation, which is the natural nontrivial choice for distance-two codes such as the Iceberg code.

\subsection{Sampling cost}
\label{sec:sampling_cost}

The total sampling cost per QED+PEC cycle consists of two contributions. The PEC variance overhead is characterized by the quasiprobability norm $\gamma_K = \sum_P |c_P^{(K)}|$, requiring an $\mathcal{O}(\gamma_K^2)$ multiplicative sample factor for bounded observables~\cite{QEM_review}; in the first-order Pauli-channel specialization used below this reduces to $\gamma_1 = 1 + 2\tilde{W}$. The QED post-selection overhead contributes a factor of $1/p_{\mathrm{success}}$. The overall sampling complexity per cycle is therefore
\begin{equation}
C_{\mathrm{cycle}} = \mathcal{O}\!\left(\frac{\gamma_K^2}{p_{\mathrm{success}}}\right).
\label{eq:sample_cost}
\end{equation}

The key observation is that QED reduces the residual noise $\tilde{W}$ substantially compared to the bare noise $W$, because most error branches are filtered out by the stabilizer measurements. When stabilizer measurements are ideal or sufficiently cheap, the post-selection penalty $1/p_{\mathrm{success}}$ can be more than compensated by the reduction in PEC overhead $\gamma_K^2$. This trade-off is the central mechanism underlying the advantage of QED+PEC over pure PEC, and is analyzed quantitatively via a toy model in Sec.~\ref{sec:toy_model}. An information-theoretic perspective based on quantum Fisher information is given in Appendix~\ref{sec:qfi}, supporting the view that subspace encoding can reduce the fundamental sampling burden under local noise.

When syndrome readout or extraction is noisy, the relevant denominator is the observed acceptance probability $p_{\mathrm{success}}^{\mathrm{obs}}$, not the ideal acceptance probability. The per-cycle cost becomes
\begin{equation}
C_{\mathrm{cycle}}^{\mathrm{obs}} = \mathcal{O}\!\left(\frac{\gamma_K^2}{p_{\mathrm{success}}^{\mathrm{obs}}}\right).
\label{eq:sample_cost_obs}
\end{equation}
Noisy extraction can therefore reverse the ideal-measurement trend if the loss in $p_{\mathrm{success}}^{\mathrm{obs}}$ exceeds the PEC saving. More generally, the total cost can become U-shaped in the detection interval $T$, with a finite optimum.

\subsection{Classical complexity of the \texorpdfstring{$K$th-order}{Kth-order} expansion}
\label{app:k_order_classical_complexity}

We bound the per-block preprocessing cost of compiling the $K$th-order QED+PEC
table. Let
\[
    m = |\mathcal K_T| = O(cnT)
\]
be the number of elementary Pauli faults in one block, where $T$ is the QEDC
interval, $c$ is the physical-layer overhead per logical layer, and $n$ is the
number of physical qubits. We treat $K$ as a fixed constant throughout.

\begin{itemize}
    \item \textbf{Number of terms}: After collecting identical Pauli channels,
    the $K$th-order expansion contains one $\mathcal P_I$ per subset
    $I\subseteq\mathcal K_T$ with $|I|\le K$, so
    \[
        N_{\mathrm{terms}}^{(K)}
        \le
        \sum_{r=0}^{K}\binom{m}{r}
        = O(m^K)
        = O\!\left((cnT)^K\right).
    \]
    We never enumerate the full $4^n$ Pauli group; only fault-generated
    strings are stored.

    \item \textbf{Propagation}: Each elementary fault is propagated to the end
    of the block in $O(n)$ via symplectic updates, giving a one-shot cost
    $O(mn)=O(cn^2T)$. A $K$-fault product is then an XOR of at most $K$
    precomputed strings, costing $O(n)$ per enumerated term:
    \[
        Cost_{\mathrm{prop}}^{(K)}
        = N_{\mathrm{terms}}^{(K)}\cdot O(n)
        = O\!\left((cnT)^K n\right).
    \]

    \item \textbf{QEDC check}: Commutation with $s$ stabilizer generators
    costs $O(sn)$ per term:
    \[
        Cost_{\mathrm{check}}^{(K)}
        = O\!\left((cnT)^K s n\right).
    \]
    Iceberg $[[k+2,k,2]]$ codes have $s=2=O(1)$ (only $X^{\otimes(k+2)}$ and
    $Z^{\otimes(k+2)}$), so this reduces to $O\!\left((cnT)^K n\right)$.

    \item \textbf{Inverse channel}: Computing
    $N_{K,\mathrm{reduce}}^{-1}=\operatorname{Trunc}_{\le K}\sum_{r=0}^{K}(-1)^r R_K^{\circ r}$
    is sparse Pauli multiplication with order truncation. The normalization
    by $P_{K,\mathrm{success}}$ and the iterated compositions $R_K^{\circ r}$
    can produce repeated fault-label powers (e.g.\ $w_i^2 w_j$), so retained
    monomials are characterized by total polynomial degree, not by distinct
    labels: every retained monomial has total degree at most $K$ in the fault
    weights, equivalently at most $K$ fault-label occurrences counting
    multiplicity. The number of such monomials is bounded by
    $\binom{m+K}{K}=O(m^K)$, and each carries an $O(n)$ Pauli multiplication, giving
    \[
        Cost_{\mathrm{inv}}^{(K)}
        = O\!\left((cnT)^K n\right).
    \]

    \item \textbf{Sampling}: Computing $\gamma_K=\sum_P|c_P^{(K)}|$ and
    $q_P^{(K)}=|c_P^{(K)}|/\gamma_K$ is a single pass over the sparse table:
    \[
        Cost_{\mathrm{sample}}^{(K)}
        = O\!\left((cnT)^K\right),
    \]
    subdominant to the steps above.
\end{itemize}

Combining the dominant terms, the per-block cost for Iceberg codes is
\[
    Cost_{\mathrm{classical}}^{(K)}
    = O\!\left((cnT)^K n\right)
    = O\!\left(c^K n^{K+1}T^K\right),
\]
and the total cost over $L/T$ blocks is $(L/T)\cdot O\!\left((cnT)^K n\right)$.
Hence the compilation is efficient in
the fixed-order sense: polynomial in $n$, $T$, and $L$ for any fixed $K$.
This is \emph{not} an all-orders statement --- if $K$ scales with $m=O(cnT)$,
the enumeration approaches the full expansion and becomes exponential in $m$.
This polynomial statement also concerns classical preprocessing only. It does not imply that the physical sampling cost remains small: noisy extraction circuits can add many fault locations and can lower the acceptance probability substantially for global stabilizer checks.

\section{Numerical Results}
\label{sec:numerical}

\subsection{Simulation Setup}
\label{sec:simulation_setup}

We compare the QED+PEC protocol with standard PEC and standard QEDC on a concrete quantum computing task: preparing a logical Greenberger--Horne--Zeilinger (GHZ) state using the $[[n, n{-}2, 2]]$ \emph{Iceberg code}~\cite{20}. This code encodes $k = n - 2$ logical qubits into $n$ physical qubits with code distance $d = 2$, and is defined by two stabilizer generators,
\begin{equation}
S_1 = Z^{\otimes n}, \qquad S_2 = X^{\otimes n}.
\label{eq:iceberg_stabilizers}
\end{equation}
The code distance $d = 2$ ensures that every single-qubit Pauli error anticommutes with at least one stabilizer and is therefore detectable by QED. However, certain weight-two errors commute with both stabilizers and remain undetected---these constitute the residual logical noise that PEC is designed to cancel.

The $n - 2$ pairs of logical operators are
\begin{equation}
\bar{Z}_j = Z_0\, Z_{j+1}, \qquad
\bar{X}_j = X_1\, X_{j+1}, \qquad
j = 1, \ldots, n-2,
\label{eq:iceberg_logicals}
\end{equation}
where subscripts label physical qubits (0-indexed). One can verify that these operators commute with both stabilizers and satisfy the canonical anticommutation relations.

A key feature of the Iceberg code is that logical CNOTs admit a simple, low-depth implementation. The logical $\overline{\mathrm{CNOT}}_{ij}$ between logical qubits $i$ and $j$ decomposes into four physical CNOTs,
\begin{equation}
\begin{split}
&\overline{\mathrm{CNOT}}_{ij} =\\
&\mathrm{CNOT}_{0,1} \circ \mathrm{CNOT}_{i+1,j+1} \circ
\mathrm{CNOT}_{0,j+1} \circ \mathrm{CNOT}_{i+1,1}
\end{split}
\label{eq:logical_cnot}
\end{equation}
which act on disjoint qubit pairs and can therefore be executed in $c = 2$ parallel layers: $\mathrm{CNOT}_{0,1}$ and $\mathrm{CNOT}_{i+1,j+1}$ in the first, $\mathrm{CNOT}_{0,j+1}$ and $\mathrm{CNOT}_{i+1,1}$ in the second.

Under the circuit-level noise model of Sec.~\ref{sec:noise_model}, each physical gate layer introduces local depolarizing noise on every active qubit or qubit pair, as illustrated in Fig.~\ref{fig:benchmark_setup}(a) for a five-qubit example.

This task prepares the logical GHZ state $\overline{|\mathrm{GHZ}\rangle} = (\overline{|0\rangle}^{\otimes(n-2)} + \overline{|1\rangle}^{\otimes(n-2)})/\sqrt{2}$ by applying $n - 3$ logical CNOT gates sequentially to the initial state $\overline{|{+}\rangle}\,\overline{|0\rangle}^{\otimes(n-3)}$. The resulting ideal circuit for $n = 10$ is shown in Fig.~\ref{fig:benchmark_setup}(b).

We adopt the gate-specific depolarizing noise model described in Sec.~\ref{sec:noise_model}, with a single-qubit error rate $p_1 = 10^{-4}$ and a two-qubit error rate $p_2 = 10^{-3}$. Stabilizer measurements are modeled as ideal projective measurements; this is a cheap-detection limit that isolates the QED+PEC mechanism. We test two departures from this assumption in Sec.~\ref{sec:syndrome_noise_checks}: a readout-only syndrome-flip model and a circuit-level noisy GHZ-assisted extraction of the Iceberg stabilizers. The term ``circuit-level noise'' in this benchmark refers to the noise on the logical-gate blocks; it does not include the gates required to extract the stabilizer syndromes unless explicitly stated.

The number of physical qubits ranges from $n = 10$ to $n = 200$, and the error detection interval $T$ (number of logical gates between consecutive QED steps) varies from 1 to 5. Throughout all simulations, we enforce the constraint $W < 0.5$, where $W = cnT\bar{p}$ is the total noise weight per QED+PEC cycle defined in Eq.~\eqref{eq:total_noise}. This keeps each block in the perturbative regime of the first-order expansion.

As one of the baselines, we compare against standard PEC applied directly to $n{-}2$ physical qubits preparing the same GHZ state without any encoding. The unencoded circuit uses fewer qubits and has shallower depth, so this comparison is, if anything, favorable to the baseline. An eight-qubit instance of the pure-PEC circuit used for this comparison is shown in Fig.~\ref{fig:benchmark_setup}(c). The sampling cost for both protocols is computed analytically from the first-order noise expansion (Sec.~\ref{sec:pec_step}), avoiding Monte Carlo statistical error.

\begin{figure}[t]
    \centering
    \includegraphics[width=\columnwidth]{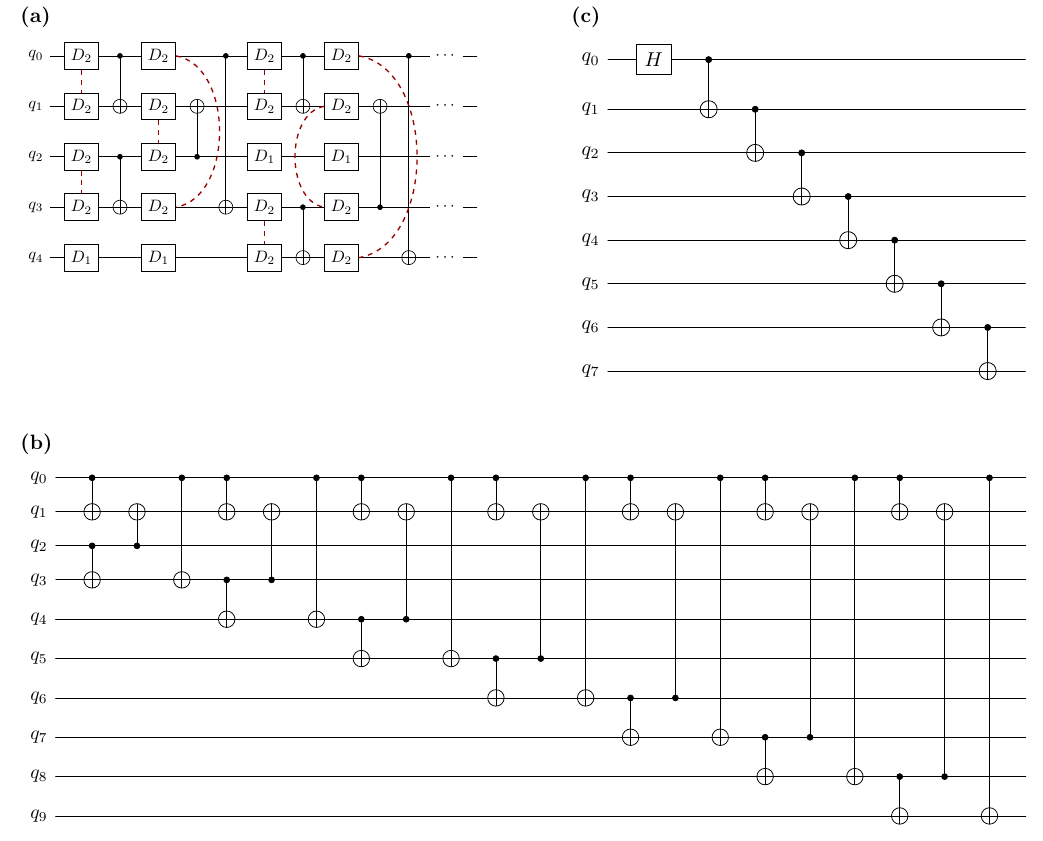}
    \caption{Benchmark setup for the Iceberg-code GHZ-state preparation task. (a) Gate-dependent circuit-level depolarizing noise model: active one- and two-qubit locations receive $D_1$ and $D_2$ channels with $p_1 = 10^{-4}$ and $p_2 = 10^{-3}$, respectively. (b) Encoded GHZ preparation on the $[[n,n{-}2,2]]$ Iceberg code. Each logical CNOT decomposes into four physical CNOT gates and is executed in two parallel layers; the panel shows the $n=10$ instance. (c) Unencoded pure-PEC baseline preparing the same GHZ state on $n{-}2$ physical qubits without error-detection post-selection.}
    \label{fig:benchmark_setup}
\end{figure}

In what follows, we show that, under ideal stabilizer measurements, QED+PEC has two advantages over the single-strategy baselines. Compared with standard PEC applied directly to the unencoded circuit, QED+PEC reduces the sampling overhead by orders of magnitude (Sec.~\ref{sec:ghz_sampling}). Compared with quantum error detection alone---i.e., post-selection on the Iceberg-code stabilizers without PEC inversion---QED+PEC achieves a higher post-selected fidelity, and this advantage grows with system size (Sec.~\ref{sec:fidelity_comparison}).

\subsection{First-order specialization}
\label{sec:pec_step}

For all numerical experiments below, we specialize the order-$K$ protocol of
Sec.~\ref{sec:high_order} to $K=1$.
This specialization should be distinguished from the general accepted-branch
set $\mathcal K_{\mathrm{success}}^{(\le K)}$ used in the arbitrary-order
construction: here we keep only the no-fault branch and accepted single-fault
branches.

Let $\mathcal K_T$ be the set of independent Pauli fault locations in one
QED+PEC block, and let
\begin{equation}
    W:=\sum_{i\in\mathcal K_T}w_i .
\end{equation}
After propagating each physical Pauli fault to the end of the Clifford block,
write the propagated Pauli as $\tilde P_i$ and the associated Pauli channel as
\[
    \mathcal P_i(\cdot)=\tilde P_i(\cdot)\tilde P_i^\dagger .
\]
Keeping terms through first order in the independent Bernoulli fault
probabilities gives
\begin{equation}
\begin{split}
    \mathcal N_{1}(\rho')
    &=
    (1-W)\rho'
    +
    \sum_{i\in\mathcal K_T}w_i\,\mathcal P_i(\rho')
    +
    O(W^2),\\
    &\hspace{2.2em}
    \rho'=\mathcal U(\rho).
\end{split}
    \label{eq:first_order_pre_qed}
\end{equation}

A single-fault branch survives the QED projection precisely when the propagated
fault commutes with every stabilizer generator.  We therefore define the
accepted single-fault set
\begin{equation}
\begin{aligned}
    \mathcal K_{\mathrm{success}}^{(1)}
    &:=
    \left\{
        i\in\mathcal K_T:\,
        [\tilde P_i,S_a]=0\right.\\
    &\qquad\left.
        \text{for all stabilizer generators }S_a
    \right\}.
\end{aligned}
    \label{eq:first_order_success_set}
\end{equation}
Rejected single faults are removed by post-selection.  Thus the first-order
success probability is
\begin{equation}
    p_{\mathrm{success}}^{(1)}
    =
    1-\sum_{i\notin\mathcal K_{\mathrm{success}}^{(1)}}w_i
    +
    O(W^2).
    \label{eq:first_order_success_prob}
\end{equation}
Conditioned on success and dropping $O(W^2)$ terms, the normalized residual
logical channel is
\begin{equation}
\begin{split}
    \mathcal{N}^{(1)}_{\mathrm{reduced}}(\rho')
    =
    \left(1-\tilde W\right)\rho'
    +
    \sum_{i\in\mathcal K_{\mathrm{success}}^{(1)}}
        \tilde w_i\,\tilde P_i\rho'\tilde P_i^\dagger \\
    \tilde w_i:=\frac{w_i}{p_{\mathrm{success}}^{(1)}},
    \qquad
    \tilde W:=\sum_{i\in\mathcal K_{\mathrm{success}}^{(1)}}\tilde w_i .
    \label{eq:N1}
\end{split}
\end{equation}
This equation makes explicit the mechanism used throughout the numerical
section: QED removes the rejected first-order branches, and PEC only inverts
the residual accepted logical noise.

Equivalently,
\begin{equation}
    \mathcal{N}^{(1)}_{\mathrm{reduced}}
    =
    \operatorname{id}+R_1,
    \qquad
    R_1
    =
    \sum_{i\in\mathcal K_{\mathrm{success}}^{(1)}}
        \tilde w_i(\mathcal P_i-\operatorname{id}) .
    \label{eq:first_order_R1}
\end{equation}
The first-order inverse is the degree-one Neumann inverse,
\begin{equation}
    \hat{\mathcal N}_1^{-1}
    =
    \operatorname{id}-R_1
    +
    O(W^2).
\end{equation}
Acting on a state, this is
\begin{equation}
    \hat{\mathcal{N}}_1^{-1}(\rho')
    =
    \left(1+\tilde W\right)\rho'
    -
    \sum_{i\in\mathcal K_{\mathrm{success}}^{(1)}}
        \tilde w_i\,\tilde P_i^\dagger\rho'\tilde P_i .
    \label{eq:N1_inv}
\end{equation}
Therefore
\begin{equation}
    \hat{\mathcal N}_1^{-1}
    \circ
    \mathcal{N}^{(1)}_{\mathrm{reduced}}
    =
    \operatorname{id}
    +
    O(W^2),
    \label{eq:first_order_identity}
\end{equation}
which is the $K=1$ instance of Theorem~\ref{thm:k_order_error_bound}.

The map in Eq.~\eqref{eq:N1_inv} is Hermiticity-preserving and
trace-preserving, and is implemented by PEC sampling.  Its quasiprobability
norm is
\begin{equation}
    \gamma_1=1+2\tilde W .
    \label{eq:first_order_gamma}
\end{equation}
The corresponding post-selection and PEC variance factors are then combined
using the general sampling-cost formula in Sec.~\ref{sec:sampling_cost}.

\subsection{Sampling Cost Comparison}
\label{sec:ghz_sampling}

Specifically, write $\gamma_k := 1 + 2\tilde{W}_k$ for the PEC variance factor of the $k$-th QED+PEC cycle [cf.~\eqref{eq:N1_inv}] and $p_{\mathrm{success},k}$ for its QED acceptance probability; the per-cycle cost~\eqref{eq:sample_cost} then composes multiplicatively across the $M = L/T$ cycles into
\begin{equation}
C_{\mathrm{QED+PEC}} \;:=\; \prod_{k=1}^{M} \frac{\gamma_k^{2}}{p_{\mathrm{success},k}},
\qquad
N \;=\; C_{\mathrm{QED+PEC}}\,/\,\varepsilon^{2},
\label{eq:C_total}
\end{equation}
For the pure-PEC baseline, we apply standard PEC directly to the unencoded
GHZ circuit on $n_{\mathrm{pure}}=n-2$ qubits.  There is no error-detection
post-selection, so the sampling overhead contains only the PEC variance
factor.  For each CNOT layer, the first-order quasiprobability norm is
\begin{equation}
\begin{split}
\gamma_{\mathrm{PEC}}
&=
1+2W_{\mathrm{PEC}}\\
W_{\mathrm{PEC}}
&=
3(n_{\mathrm{pure}}-2)w_1+15w_2
=
3(n-4)w_1+15w_2
\end{split}
\end{equation}
Since the unencoded GHZ circuit contains $n_{\mathrm{pure}}-1=n-3$ CNOT
layers, the total pure-PEC sampling overhead is
\begin{equation}
\begin{split}
C_{\mathrm{PEC}}
&=
\gamma_{\mathrm{PEC}}^{\,2(n_{\mathrm{pure}}-1)}
=
\left[1+2\bigl(3(n-4)w_1+15w_2\bigr)\right]^{2(n-3)}\\
N_{\mathrm{PEC}}
&=
C_{\mathrm{PEC}}/\varepsilon^2
\end{split}
\end{equation}

Figure~\ref{fig:main_ideal_result}(a) compares the total sampling overhead of QED+PEC and standard PEC as a function of the number of physical qubits $n$ on a logarithmic scale. Note that the sampling cost reported for our QED+PEC protocol accounts not only for the PEC variance overhead but also for the repetitions incurred by discarding runs that fail QED post-selection. We emphasize that the vertical axis in panel~(a) shows the \emph{precision-independent} part of the sampling cost---i.e., the variance amplification factor $\prod_k \gamma_k^2 / p_{\mathrm{success},k}$---rather than the absolute number of samples required to achieve a given estimation accuracy~$\varepsilon$. The full sample count is $N = \varepsilon^{-2}\prod_k \gamma_k^2 / p_{\mathrm{success},k}$, but since the $\varepsilon^{-2}$ prefactor is common to all protocols, it cancels in the ratio plotted in panel~(b). Consequently, the absolute values in panel~(a) reflect relative trends across protocols, while the ratios in panel~(b) are numerically exact.

All methods exhibit at least exponential growth in sample cost with $n$, as the total noise weight $W$ accumulates over the full circuit depth. At $n = 100$, taking $\epsilon = 1$, standard PEC already requires $\sim 60$ samples per estimate; at $n = 200$, this explodes to $\sim 10^7$. In contrast, the QED+PEC curves grow far more slowly. For $T = 1$ (error detection after every logical gate), the sampling cost at $n = 200$ is only about $10^3$---an improvement of roughly three to four orders of magnitude over standard PEC. This does not violate the exponential lower bounds for unbiased QEM~\cite{16,17,18}: ideal projection shrinks the residual noise volume entering the PEC variance factor, while the exponential growth remains in the post-selection denominator $\prod_k p_{\mathrm{success},k}^{-1}$. The curves for different $T$ values are clearly stratified: in this ideal-measurement benchmark, more frequent error detection (smaller $T$) yields lower overhead, and the gap between protocols widens with increasing $n$.

Within the ideal-measurement model, this behavior is a direct manifestation of the trade-off analyzed in Secs.~\ref{sec:sampling_cost} and~\ref{sec:toy_model}: error detection filters out the majority of noise branches via post-selection, keeping the residual noise weight $\tilde{W}$ per cycle small and thereby suppressing the PEC variance factor $\gamma^2 = (1 + 2\tilde{W})^2$. Although QED post-selection introduces additional repetitions due to discarded runs, the resulting reduction in PEC sampling overhead more than compensates, leading to an exponential net decrease in the total sampling cost.

\begin{figure}[htbp]
    \centering
    \includegraphics[width=\textwidth]{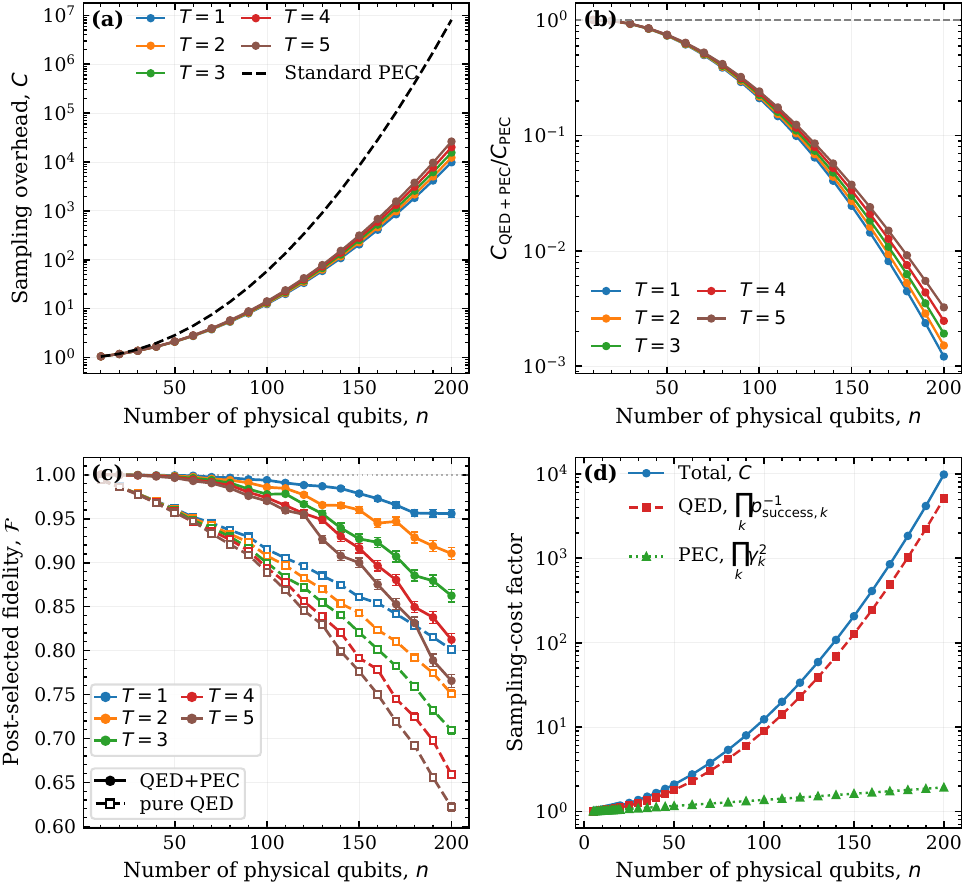}
    \caption{Main ideal-measurement benchmark for logical GHZ-state preparation on the $[[n,n{-}2,2]]$ Iceberg code with $p_1=10^{-4}$ and $p_2=10^{-3}$, assuming ideal stabilizer measurements. (a) Total sampling overhead $C$ for QED+PEC at detection intervals $T=1,\ldots,5$ compared with standard PEC on $n{-}2$ unencoded qubits. (b) Overhead ratio $C_{\mathrm{QED+PEC}}/C_{\mathrm{PEC}}$, showing the orders-of-magnitude reduction at large $n$. (c) Post-selected logical GHZ fidelity of QED+PEC compared with pure QED under the same encoding and noise model. Error bars show one Monte Carlo standard error. (d) Decomposition of the $T=1$ QED+PEC cost into the total sampling factor, the QED post-selection factor, and the residual PEC variance factor. Together, these panels show that ideal QED projections reduce the residual channel that PEC must invert.}
    \label{fig:main_ideal_result}
\end{figure}

Figure~\ref{fig:main_ideal_result}(b) plots the ratio $C_{\text{QED+PEC}} / C_{\text{PEC}}$ as a function of $n$. As $n$ increases, the ideal-measurement advantage of QED+PEC over standard PEC grows rapidly: the ratio decreases on a logarithmic scale at a rate that is at least linear, consistent with an exponential reduction in the reported sampling-cost factor over the simulated range. Moreover, the curves are clearly stratified by the detection frequency $T$: smaller $T$ (more frequent error detection) consistently yields a lower sampling cost ratio, demonstrating that increasing the detection frequency further suppresses the total overhead in this regime.

To understand the origin of this advantage, Fig.~\ref{fig:main_ideal_result}(d) decomposes the total sampling cost $C_{\mathrm{QED+PEC}}$ for the $T=1$ case into the QED post-selection overhead $\prod_k p_{\mathrm{success},k}^{-1}$ and the PEC variance overhead $\prod_k \gamma_k^{2}$. The QED post-selection factor dominates the growth with $n$, while the PEC variance factor increases only slowly. This confirms that, under first-order PEC and ideal stabilizer measurements, the residual noise after QED is small and contributes modestly to the total sampling cost. Since each error event canceled by ideal QED costs far fewer samples than one canceled by PEC, offloading as much of the noise budget as possible onto QED is the key driver of the overhead reduction observed in Fig.~\ref{fig:main_ideal_result}(a,b).

\subsection{Fidelity Verification and Comparison with Pure Error Detection}
\label{sec:fidelity_comparison}
To benchmark the performance of our protocol, we compute the fidelity between the post-selected output of the QED+PEC protocol and the ideal noiseless logical GHZ state evolved under the same circuit,
\begin{equation}
\begin{split}
\mathcal{F} &= \mathcal{F}(\left(\hat{N}^{-1}_1\circ N_{\text{reduced}} \circ \mathcal{U}\right)^{(L/T)}\circ \cdots \\
&\circ \left(\hat{N}^{-1}_1\circ N_{\text{reduced}} \circ \mathcal{U}\right)^{(1)}(\rho)\ ,\ \mathcal{U}^{(L/T)}\cdots
\mathcal{U}^{(1)}(\rho))
\end{split}
\label{eq:fidelity_def_qedpec}
\end{equation}
Here this quantity is computed statistically as $\Tr(\op{GHZ}{GHZ}\rho_{\text{mitigated}})$. The fidelity is estimated by Monte Carlo sampling of the PEC quasiprobability distribution, a target standard error of $10^{-3}$, and a minimum of $10{,}000$ post-selected samples per data point. We emphasize that the inversion map $\hat{N}^{-1}_1$ retains only the leading-order term of the PEC expansion; the results reported below therefore reflect what is already achievable with a first-order correction.

For reference, we evaluate the analogous fidelity for the pure QEDC protocol, obtained by replacing each $\hat{N}^{-1}_1\circ N_{\text{reduced}}$ block by $N_{\text{reduced}}$ alone, i.e., by removing the PEC inversion step:
\begin{equation}
\begin{split}
\mathcal{F} &= \mathcal{F}(\left( N_{\text{reduced}} \circ \mathcal{U}\right)^{(L/T)}\circ \cdots \\
&\circ \left( N_{\text{reduced}} \circ \mathcal{U}\right)^{(1)}(\rho)\ ,\ \mathcal{U}^{(L/T)}\cdots\mathcal{U}^{(1)}(\rho))
\end{split}
\label{eq:fidelity_def_qedc}
\end{equation}
The same data also provide a numerical check of the perturbative scale implied by the error bound. For the first-order implementation, Corollary~\ref{cor:perturbative_scaling} and Eq.~\eqref{eq:end_to_end_algorithmic_bound_main} motivate the end-to-end fidelity-loss scale
\[
B_1 := \exp\!\left(\sum_m W_m^2\right)-1 ,
\]
up to the certified constants in the one-block quantities $\epsilon_{m,1}$. Since the GHZ projector has operator norm one, Eq.~\eqref{eq:observable_bias_algorithmic_bound_main} applies directly to $1-\mathcal{F}$. Representative data points are shown in Table~\ref{tab:fidelity_bound_check}. This comparison shows that the theorem-motivated perturbative scale is informative at small $B_1$ and becomes numerically loose at larger values, as expected for a worst-case diamond-norm upper bound applied to a structured observable.

\begin{table}[htbp]
    \centering
    \caption{Comparison between the first-order perturbative scale and the observed GHZ fidelity loss for representative error-detection intervals.}
    \label{tab:fidelity_bound_check}
    \begin{tabular}{ccc}
        \toprule
        Parameter $(n,T)$ & $B_1$ & $1-\mathcal F$ \\
        \midrule
        $(30,1)$  & $2.29\times 10^{-3}$ & $2.18\times 10^{-4}$ \\
        $(100,1)$ & $5.36\times 10^{-2}$ & $6.09\times 10^{-3}$ \\
        $(200,1)$ & $4.44\times 10^{-1}$ & $4.42\times 10^{-2}$ \\
        $(30,5)$  & $1.10\times 10^{-2}$ & $5.54\times 10^{-4}$ \\
        $(100,5)$ & $2.94\times 10^{-1}$ & $2.98\times 10^{-2}$ \\
        $(200,5)$ & $5.22$                & $2.34\times 10^{-1}$ \\
        \bottomrule
    \end{tabular}
\end{table}

Figure~\ref{fig:main_ideal_result}(c) reports the post-selected fidelity of QED+PEC (solid markers) and pure QEDC (open markers, dashed) for $T = 1, \ldots, 5$ and $n = 10, \ldots, 200$. Across the entire ideal-measurement parameter range, QED+PEC outperforms pure QEDC by a wide margin. The reduction is most clearly seen on the infidelity scale: at $T = 1$, the infidelity $1 - \mathcal{F}$ is suppressed from $2.2 \times 10^{-2}$ to $2.2 \times 10^{-4}$ at $n = 30$, from $8.5 \times 10^{-2}$ to $6.1 \times 10^{-3}$ at $n = 100$, and from $2.0 \times 10^{-1}$ to $4.4 \times 10^{-2}$ at $n = 200$, corresponding to roughly two orders of magnitude of suppression for small to moderate systems and a factor of approximately $4$ to $5$ at $n = 200$. The same comparison can be phrased in terms of system-size scaling at a fixed fidelity target: to maintain $\mathcal{F} > 0.95$ at $T = 1$, pure QEDC is limited to $n \lesssim 65$, whereas QED+PEC remains above this threshold over the full range tested up to $n = 200$; for the more demanding setting $T = 5$, the corresponding cutoffs are $n \approx 55$ for pure QEDC and $n \approx 125$ for QED+PEC, a factor of roughly $2.3$ in supported logical size. The advantage is consistent across every detection interval studied and is achieved using only the first-order PEC expansion, indicating that even a leading-order quasiprobability correction substantially suppresses the residual logical noise that escapes the error-detection step.

\subsection{Robustness to Characterization Uncertainty}
\label{sec:characterization_robustness}

This test concerns sensitivity to uncertainty in an already-available physical gate-noise model; it is not an estimate of the sample complexity required to learn that model, nor does it model syndrome-extraction circuit noise. The first-order PEC correction table uses the noise parameters $p_1$ and $p_2$ as fixed inputs. In practice these are estimated from characterization experiments and carry finite statistical uncertainty. To test sensitivity to such errors, we perturb every elementary fault probability from its nominal value $p$ to $p\,(1+r)$, where $r$ is drawn independently from $\mathrm{Uniform}[-r_{\max}, r_{\max}]$ for each fault channel, while the PEC table continues to use the nominal $p_1$, $p_2$. The maximum relative drift $r_{\max}$ is swept from $5\times10^{-4}$ to $0.5$, with $n=10$--$200$ and $T=1$--$5$.

\begin{figure}[t]
    \centering
    \includegraphics[width=\columnwidth]{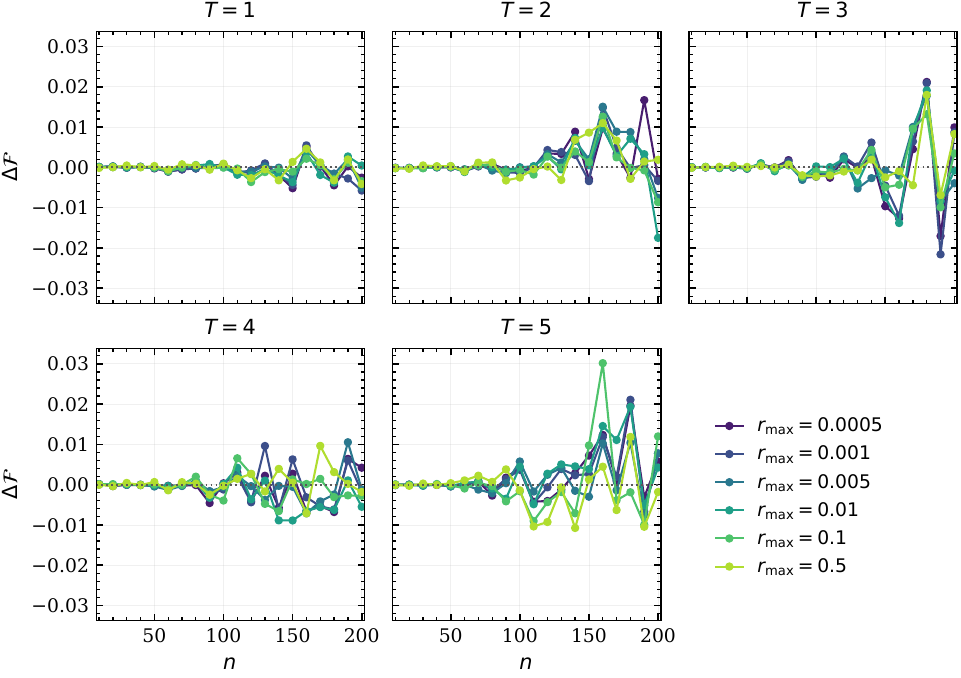}
    \caption{Fidelity impact of noise-characterization drift. Each panel corresponds to a fixed detection interval $T$ and shows the fidelity difference $\Delta\mathcal{F} = \mathcal{F}_{\mathrm{drifted}} - \mathcal{F}_{\mathrm{baseline}}$ versus $n$ for six values of $r_{\max}$ (see legend). The dashed line marks $\Delta\mathcal{F} = 0$. The mean $\Delta\mathcal{F}$ remains at the $10^{-4}$ level for all $r_{\max}$, and fluctuations stay within the Monte Carlo standard error of the baseline (${\sim}10^{-3}$).}
    \label{fig:rmax_delta}
\end{figure}

Figure~\ref{fig:rmax_delta} reports the fidelity difference across all parameter combinations. For every $r_{\max}$ value tested, the mean $\Delta\mathcal{F}$ is at the $10^{-4}$ level, and the point-by-point fluctuations around zero lie within the Monte Carlo standard error of the baseline simulation. Even at $r_{\max}=0.5$, where individual fault probabilities can be perturbed by up to $\pm 50\%$, the aggregate fidelity impact is statistically indistinguishable from zero. This insensitivity follows from the structure of the first-order PEC inversion: the correction depends on the per-block total noise weight $W = cnT\bar{p}$, which averages over many independent fault channels, so random perturbations of individual channels largely cancel. The practical implication is that the QED+PEC protocol does not require high-precision characterization of each elementary noise rate---knowledge of the average depolarizing parameters at the ${\sim}10\%$ level suffices for this first-order ideal-measurement implementation.

\subsection{Syndrome-noise checks}
\label{sec:syndrome_noise_checks}

The benchmark above assumes ideal stabilizer measurements. We now examine two departures from this assumption: a classical readout-only flip model and a circuit-level noisy GHZ-assisted extraction model.

\subsubsection{Readout-only syndrome flips}
\label{sec:readout_only}

We first isolate the effect of classical syndrome-record errors. After each otherwise ideal stabilizer measurement, each reported outcome is flipped independently with probability $p_m$. The first-order PEC table is unchanged. This model is a diagnostic that separates false rejection from false acceptance; it is not a full extraction model for the global Iceberg checks.

Figure~\ref{fig:readout_fidelity} shows the final GHZ fidelity for $p_m$ from $10^{-5}$ to $5\times10^{-3}$, $n=50,100,150,200$, and $T=1,\ldots,5$. No systematic monotonic fidelity loss with $p_m$ is observed. The dominant dependence is on $n$ and $T$. For example, at $n=200$ and $T=1$, the fidelity varies from $0.946$ to $0.957$ over the whole $p_m$ sweep, with one-sigma Monte Carlo error of order $5\times10^{-3}$. Thus readout-only errors do not visibly change the first-order logical correction accuracy in this test.

\begin{figure}[t]
    \centering
    \includegraphics[width=\columnwidth]{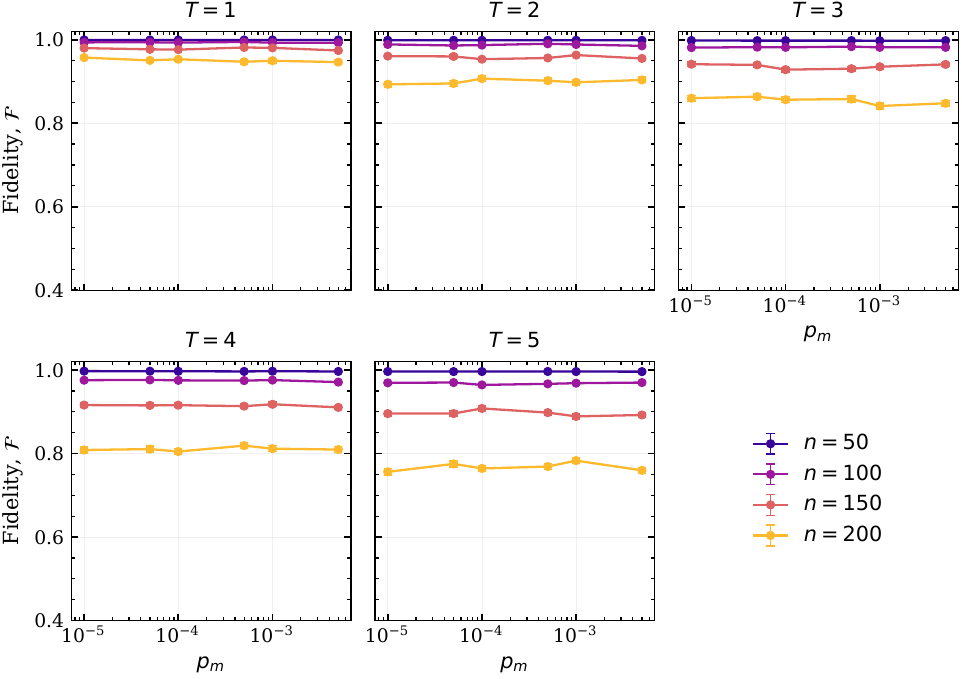}
    \caption{Effect of readout-only syndrome errors on the QED+PEC GHZ fidelity. Each reported stabilizer outcome is flipped independently with probability $p_m$ after an otherwise ideal stabilizer measurement, and the PEC table is the same first-order table used in the ideal-syndrome calculation. No systematic monotonic fidelity loss is observed over $p_m=10^{-5}$ to $5\times10^{-3}$; the main dependence is on $n$ and $T$. Error bars show one Monte Carlo standard error.}
    \label{fig:readout_fidelity}
\end{figure}

The sampling cost behaves differently. Let $p_{\mathrm{success}}^{\mathrm{true}}$ be the probability that the underlying ideal syndrome is trivial, and $p_{\mathrm{success}}^{\mathrm{obs}}$ the probability that the reported syndrome is accepted after readout flips. The cost is controlled by $p_{\mathrm{success}}^{\mathrm{obs}}$. Figure~\ref{fig:readout_sample_cost} shows that readout flips increase the observed post-selection factor while leaving the PEC variance factor nearly unchanged. For $p_m=10^{-3}$ and $T=1$, the total cost increases by factors of $1.10$, $1.21$, and $1.47$ at $n=50,100,200$, respectively. This is consistent with the first-order mechanism in Sec.~\ref{sec:ed_step}: a single readout flip can falsely reject an accepted trajectory at first order in $p_m$, whereas false acceptance of a detectable fault is second order.

\begin{figure}[t]
    \centering
    \includegraphics[width=\columnwidth]{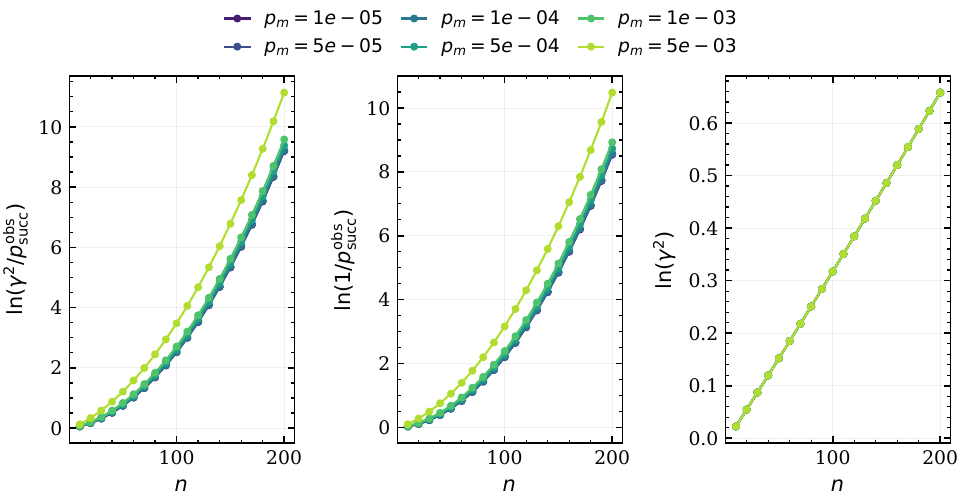}
    \caption{Sampling-cost effect of readout-only syndrome errors for $T=1$. Left: observed total cost factor $\ln(\gamma^2/p_{\mathrm{success}}^{\mathrm{obs}})$. Middle: observed post-selection factor $\ln(1/p_{\mathrm{success}}^{\mathrm{obs}})$. Right: PEC variance factor $\ln(\gamma^2)$. Readout errors mainly increase the post-selection denominator, while the PEC factor remains almost unchanged.}
    \label{fig:readout_sample_cost}
\end{figure}

\subsubsection{Circuit-level noisy GHZ-assisted extraction}
\label{sec:noisy_ghz_extraction}

The readout-only model above does not capture the main physical cost of measuring the Iceberg stabilizers, because $X^{\otimes n}$ and $Z^{\otimes n}$ are global checks. We therefore also test an explicit GHZ-assisted extraction circuit. For each QED round, ideal GHZ ancillas of size $m$ are supplied for the two global stabilizers. Data--ancilla CNOTs are scheduled in parallel batches; CNOTs sharing a data qubit are delayed to the next available layer. The same one- and two-qubit depolarizing noise rates, $p_1=10^{-4}$ and $p_2=10^{-3}$, are applied to the extraction circuit, including ancilla operations. The data below use $m=n/2$ and one extraction layer per QED round.

Figure~\ref{fig:noisy_ghz_extraction_limit}(a) shows that the circuit-level extraction cost changes the conclusion qualitatively. At $n=80$, the noisy-extraction cost is about $114.8$, compared with $6.71$ for ideal measurement and $13.82$ for pure PEC; the accepted-shot product drops from $0.865$ at $n=20$ to $0.074$ at $n=80$. For the Iceberg code, each stabilizer is a weight-$n$ global check, so any CNOT-based extraction introduces $\Theta(n)$ additional noisy two-qubit interactions per detection round. Increasing the detection frequency therefore also multiplies the number of noisy extraction circuits.

\begin{figure}[t]
    \centering
    \begin{minipage}[t]{0.49\textwidth}
        \textbf{\normalsize (a)}\\[-0.5ex]
        \includegraphics[width=\textwidth]{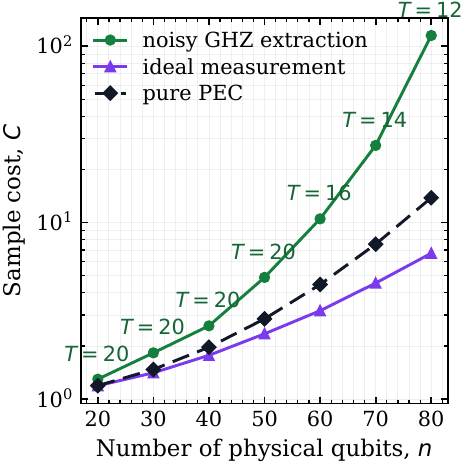}
    \end{minipage}\hfill
    \begin{minipage}[t]{0.49\textwidth}
        \textbf{\normalsize (b)}\\[-0.5ex]
        \includegraphics[width=\textwidth]{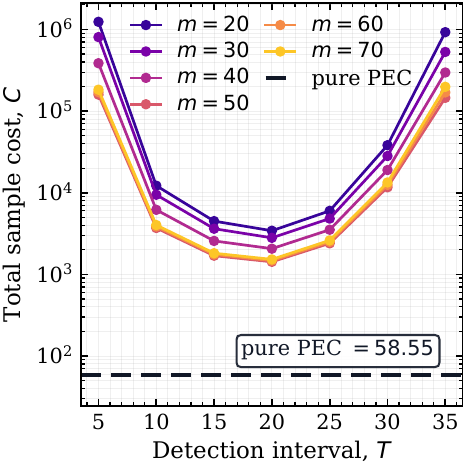}
    \end{minipage}
    \caption{Noisy GHZ-assisted extraction as a boundary condition for the Iceberg-code QED+PEC advantage. (a) Circuit-level noisy GHZ-assisted extraction can erase the ideal-measurement advantage: the noisy-extraction QED+PEC cost grows above both the ideal stabilizer-measurement baseline and pure PEC for the tested global-check Iceberg extraction circuit. (b) At fixed $n=100$, noisy extraction changes the detection-frequency dependence from monotonic improvement at small $T$ to a finite-$T$ optimum. The dashed horizontal line in (b) is the pure-PEC baseline. Together, the panels show that the ideal stabilizer-projection advantage survives only when syndrome extraction is sufficiently cheap.}
    \label{fig:noisy_ghz_extraction_limit}
\end{figure}

To test the detection-interval dependence directly, we fix $n=100$ and sweep $T$ under the noisy GHZ-assisted extraction model. Figure~\ref{fig:noisy_ghz_extraction_limit}(b) shows the total cost for $m=20,\ldots,70$, together with the pure-PEC baseline $C_{\mathrm{PEC}}=58.55$. All curves are U-shaped in $T$ and reach their minimum near $T=20$, not at the smallest tested interval. For $m=50$ the cost decreases from $1.59\times10^{5}$ at $T=5$ to $1.42\times10^{3}$ at $T=20$, and then rises again to $1.45\times10^{5}$ at $T=35$. Even at the best noisy-extraction point, the cost is about $24.3$ times the pure-PEC baseline in this setting. The ideal-measurement rule that smaller $T$ gives lower cost therefore does not transfer to circuit-level noisy extraction; instead, smaller $T$ repeats the noisy global-check extraction too often, while larger $T$ leaves a larger residual channel for PEC.

The optimal detection interval is code-, circuit-, and hardware-dependent. The negative result here applies specifically to noisy CNOT-based extraction of the global Iceberg checks; it does not show that the perturbative QED+PEC construction fails in general. Codes with low-weight checks, subsystem codes, qLDPC-like codes, or hardware-native parity measurements are natural candidates in which the extraction overhead may be small enough for the QED+PEC trade-off to remain favorable.

\section{Physical Intuition from Toy Models}
\label{sec:toy_model}

Our ideal-measurement numerical results show that the total sample complexity of the
QED+PEC protocol \emph{decreases} as the error-detection frequency
increases.
This is a nontrivial observation: more frequent post-selection
discards more shots, yet the net effect is a reduction in overall cost.
To understand why, we construct two analytically solvable toy models
and compare their scaling with the detection/cancellation interval
$\tau$.
The toy models below intentionally exclude syndrome-extraction circuit noise; they explain the cheap-detection mechanism, not the full hardware trade-off once global stabilizer extraction is noisy.

\subsection{Two toy models}
\label{sec:toy:setup}

Both models share the same depolarizing noise rate $\gamma$ and total
evolution time $T$, but differ in the Hilbert space structure and the
error-mitigation strategy.

\paragraph{Model A (QED+PEC).}
The system lives in an $N=2^n$-dimensional Hilbert space containing a
$2$-dimensional code subspace with projector $P$
($\mathrm{Tr}\,P=2$).
The initial state $\rho_0$ is supported on the code subspace.
Noise is global depolarization at rate $\gamma$:
after a time interval $\tau$, the state becomes
\begin{equation}
\rho(\tau)=(1-p)\,\rho_0+\frac{p}{N}\,I\,,
\qquad p=1-e^{-\gamma\tau}\,.
\label{eq:toy_rho}
\end{equation}
At the end of each interval, we project onto $P$ (error detection) and
discard runs that fail, then apply PEC to invert the residual noise
within the code subspace.

\paragraph{Model B (Pure PEC, no ancilla).}
{
The system is unencoded, confined to a $2$-dimensional physical qubit space.}
Noise is depolarization within this space:
\begin{equation}
\rho(\tau)=(1-p)\,\rho_0+p\,\frac{I}{2}\,,
\qquad p=1-e^{-\gamma\tau}\,.
\label{eq:toy_rhoB}
\end{equation}
Since the system never leaves the {physical space}, there is no
post-selection ($S_1=1$); the only mitigation tool is PEC, applied
every $\tau$.

{This toy model relies on simplifying assumptions (isotropic depolarization,
ideal measurements, single logical qubit) that are not hardware-faithful,
but they suffice to capture the essential trade-off between post-selection
cost and PEC variance reduction.
The qualitative prediction below---that frequent error detection reduces sample
complexity---is confirmed by the full numerical simulations in
Sec.~\ref{sec:numerical}.}

\subsection{Derivation and core scaling}
\label{sec:toy:results}

We now derive the leading sample-complexity scaling for the two models.  The
appendix keeps only the intermediate Taylor-expansion checks.

\paragraph{Model A (QED+PEC).}
After one interval, projection onto the two-dimensional code subspace gives
\begin{equation}
    \tilde\rho_P(\tau)
    =
    P\rho(\tau)P
    =
    (1-p)\rho_0+\frac{p}{N}P ,
    \qquad
    p=1-e^{-\gamma\tau}.
    \label{eq:toyA_unnormalized}
\end{equation}
Hence the single-round survival probability is
\begin{equation}
    p_{\mathrm{succ}}
    =
    \operatorname{Tr}\tilde\rho_P(\tau)
    =
    1-p\left(1-\frac{2}{N}\right).
    \label{eq:toyA_psucc}
\end{equation}
Within the code subspace, the normalized post-selected state is a depolarizing
Pauli channel acting on $\rho_0$.  Writing
\[
    P=\frac{1}{2}\left(\rho_0+X\rho_0X+Y\rho_0Y+Z\rho_0Z\right),
\]
the normalized channel has identity coefficient
\[
    A=\frac{1-p+p/(2N)}{p_{\mathrm{succ}}}
\]
and equal nonidentity Pauli coefficients
\[
    B=\frac{p}{2Np_{\mathrm{succ}}}.
\]
The nontrivial Pauli-transfer eigenvalue is therefore
\begin{equation}
    \lambda_A=A-B
    =
    \frac{1-p}{1-p(1-2/N)} .
    \label{eq:toyA_lambda}
\end{equation}
The single-round PEC norm for the inverse of an isotropic qubit
depolarizing channel is
\begin{equation}
    \gamma_{\mathrm{PEC}}^{(A)}
    =
    \frac{3}{2}\lambda_A^{-1}-\frac{1}{2}
    =
    1+\frac{3}{N}\left(e^{\gamma\tau}-1\right).
    \label{eq:toyA_gamma}
\end{equation}
Over $T/\tau$ rounds, the post-selection and PEC variance factors are
\begin{equation}
    S_1=p_{\mathrm{succ}}^{-T/\tau},
    \qquad
    S_2=\left[(\gamma_{\mathrm{PEC}}^{(A)})^2\right]^{T/\tau}.
    \label{eq:toyA_S1S2}
\end{equation}
Expanding the logarithms through the order that contributes
$O(\gamma^2T\tau)$ to the total exponent gives
\begin{align}
    \ln S_1
    &=
    \gamma T\left(1-\frac{2}{N}\right)
    +
    \gamma^2T\tau\left(-\frac{1}{N}+\frac{2}{N^2}\right),
    \label{eq:toyA_lnS1}\\
    \ln S_2
    &=
    \frac{6\gamma T}{N}
    +
    \gamma^2T\tau\left(\frac{3}{N}-\frac{9}{N^2}\right).
    \label{eq:toyA_lnS2}
\end{align}
Thus
\begin{equation}
\begin{split}
    \ln &C_{\mathrm{QED{+}PEC}}
    =
    \ln(S_1S_2)\\
    &=
    \gamma T\left(1+\frac{4}{N}\right)
    +
    \gamma^2T\tau\left(\frac{2}{N}-\frac{7}{N^2}\right).
    \end{split}
    \label{eq:toy_totalA}
\end{equation}

\paragraph{Model B (Pure PEC).}
For the unencoded two-dimensional system, there is no post-selection penalty.
The depolarizing channel over one interval can be written as
\begin{equation}
\begin{split}
    \mathcal E_B(\rho_0)
    &=
    \left(1-\frac{3p}{4}\right)\rho_0
    +
    \frac{p}{4}
    \left(
        X\rho_0X+Y\rho_0Y+Z\rho_0Z
    \right),\\
    &\hspace{2.2em}
    p=1-e^{-\gamma\tau}.
\end{split}
    \label{eq:toyB_channel}
\end{equation}
Its nontrivial Pauli-transfer eigenvalue is
$\lambda_B=1-p=e^{-\gamma\tau}$, so the single-round PEC norm is
\begin{equation}
    \gamma_{\mathrm{PEC}}^{(B)}
    =
    \frac{3}{2}e^{\gamma\tau}-\frac{1}{2}.
    \label{eq:toyB_gamma}
\end{equation}
Applying PEC every $\tau$ over the full duration $T$ gives
\begin{equation}
\begin{split}
    \ln C_{\mathrm{pure\;PEC}}^{(\tau)}
    =
    \frac{T}{\tau}
    \ln\left[
        \left(
            \frac{3}{2}e^{\gamma\tau}-\frac{1}{2}
        \right)^2
    \right]\\
    =
    3\gamma T-\frac{3}{4}\gamma^2T\tau+O(\gamma^3T\tau^2).
    \end{split}
    \label{eq:toy_totalB}
\end{equation}
The negative $O(\gamma^2T\tau)$ correction means that periodic pure PEC
becomes cheaper as $\tau$ increases.  The optimal pure-PEC strategy in this
model is therefore the single-shot choice $\tau=T$, whose exact cost is
\begin{equation}
    C_{\mathrm{pure\;PEC}}^{(\tau=T)}
    =
    \left(
        \frac{3}{2}e^{\gamma T}-\frac{1}{2}
    \right)^2
    \approx
    \frac{9}{4}e^{2\gamma T}
    \qquad
    (\gamma T\gg 1).
    \label{eq:toy_purePEC_bound}
\end{equation}

\subsection{Physical interpretation}
\label{sec:toy:interpretation}

The contrast between Eqs.~\eqref{eq:toy_totalA} and
\eqref{eq:toy_totalB} is striking: the $\mathcal{O}(\gamma^2 T\tau)$
correction has \emph{opposite signs} in the two models.

\paragraph{Pure PEC prefers infrequent cancellation.}
In Model~B, the correction $-\frac{3}{4}\gamma^2 T\tau$ is negative:
increasing $\tau$ (fewer, larger PEC steps) \emph{reduces} the total
cost.
This is because splitting a noise channel into many small pieces and
inverting each one separately compounds the quasiprobability overhead
multiplicatively, making frequent PEC strictly
suboptimal~\cite{31}.
The best pure-PEC strategy is therefore to accumulate all noise and
cancel it in a single shot ($\tau=T$).

\paragraph{QED+PEC prefers frequent detection in the ideal-projection model.}
In Model~A, the correction $+\gamma^2 T\tau\,(2/N-7/N^2)$ is positive
for $N\ge 4$: increasing $\tau$ \emph{increases} the total cost.
As $\tau$ decreases, the total post-selection cost $S_1$ grows because
more rounds of projection discard more shots.
At the same time, the total PEC variance overhead $S_2$ shrinks because
each projection filters out the leaked population, leaving a weaker
effective noise channel within the code subspace that is cheaper to
invert.
Crucially, $S_2$ decreases faster than $S_1$ increases, so the product
$C=S_1 S_2$ drops with decreasing $\tau$---a discrete analogue of the
quantum Zeno effect.

\paragraph{Comparison: the advantage of subspace structure.}
A direct comparison of the two models reveals the advantage conferred
by the ancilla subspace and error detection.
Even the \emph{optimal} pure-PEC strategy (single-shot cancellation,
$\tau=T$) has a sample complexity scaling as $\sim e^{2\gamma T}$
[Eq.~\eqref{eq:toy_purePEC_bound}], whereas the QED+PEC cost in the
frequent-detection limit ($\tau\to 0$) scales as
$e^{\gamma T(1+4/N)}$, with an exponential coefficient
$1+4/N < 2$ for all $N>4$.
In this ideal-projection model, this exponential gap shows that subspace filtering can provide a qualitative advantage that pure PEC cannot achieve by changing the schedule alone. The noisy-extraction data in Sec.~\ref{sec:syndrome_noise_checks} show that this advantage survives in practice only when the physical cost of syndrome extraction is sufficiently small.

\section{Discussion and Outlook}
\label{sec:discussion}

\subsection{Summary}
\label{sec:discussion_summary}

The central regime targeted by this work is the intermediate territory between NISQ and early FTQC: pure PEC suffers from a sampling overhead that grows exponentially in the noise-weighted circuit volume, while pure QEDC is limited by an irreducible logical bias from undetectable fault events and an acceptance probability that decays with the circuit size. These two limitations are structurally distinct, and our main thesis is that they should be addressed together. Error detection and probabilistic error cancellation are not redundant techniques aimed at the same noise: the former filters out the bulk of detectable faults at the level of the stabilizer subspace, while the latter cancels the residual undetectable noise that survives post-selection. Combining the two yields a feedback-free protocol whose joint cost-versus-bias performance can dominate either component alone in a cheap-detection regime.

\subsection{Main contributions}
\label{sec:discussion_contrib}

We constructed a feedback-free QED+PEC protocol that interleaves blocks of noisy logical computation with QEDC stabilizer measurements and PEC inversion on the surviving logical channel. The protocol requires neither real-time syndrome decoding nor active recovery operations: failed trajectories are simply discarded, either online or in classical post-processing.

Conceptually, the role of QED is not only to discard bad trajectories, but to reshape the effective logical channel that PEC must invert: detectable fault branches are removed by stabilizer post-selection, while the surviving branches are renormalized into a weaker residual channel on which quasiprobability cancellation is applied.

A central technical step is the perturbative inverse construction of the post-selected residual channel. For any fixed perturbative order $K$, the number of fault branches that must be enumerated is $O(m^K)$ rather than $2^m$, so that the classical preprocessing remains polynomial in the system size $n$, the detection interval $T$, and the circuit depth. We proved that the order-$K$ implementation cancels the normalized post-selected channel through total fault order $K$, leaving an $O(W^{K+1})$ worst-case residual per QED+PEC block under the assumptions of Theorem~\ref{thm:k_order_error_bound} and Corollary~\ref{cor:perturbative_scaling}, with linear-order accumulation only when $\sum_m\epsilon_{m,K}\ll1$. We further decomposed the total sampling cost into a post-selection penalty and a PEC quasiprobability variance factor, making explicit how QEDC reduces the effective PEC overhead by shrinking the post-selected residual noise weight $\tilde{W}$.

\subsection{Numerical findings}
\label{sec:discussion_numerics}

Under ideal stabilizer measurements, the logical GHZ-preparation benchmark using the $[[n,n{-}2,2]]$ Iceberg code shows a large QED+PEC advantage. At $n=200$ and $T=1$, the first-order protocol reduces the sampling overhead by roughly three to four orders of magnitude relative to standard unencoded PEC while maintaining $\mathcal{F}\simeq 0.956$. Compared with pure QEDC under the same ideal-measurement assumption, QED+PEC gives a higher post-selected fidelity, confirming that PEC cancels part of the undetectable logical noise left by QED. In this regime, increasing the detection frequency simultaneously improves the fidelity and reduces the total sampling cost, in contrast to pure PEC, where splitting the noise channel into many small steps compounds the quasiprobability overhead~\cite{31}.

The syndrome-noise checks show that this ideal-measurement trend does not extrapolate directly to noisy syndrome extraction. Readout-only flips leave the fidelity statistically unchanged but increase the observed post-selection penalty. A circuit-level noisy GHZ-assisted extraction of the global Iceberg stabilizers can make the QED+PEC sampling cost much larger than both the ideal-measurement value and pure PEC. A fixed-$n$ sweep further turns the monotonic preference for small $T$ into a finite-$T$ optimization problem, with the optimum near $T=20$ in our setting. The numerical message is therefore two-sided: the protocol can have a large cheap-detection advantage, but this advantage can be lost when syndrome extraction is sufficiently noisy.

\subsection{Physical interpretation}
\label{sec:discussion_physics}

The ideal-measurement QED+PEC advantage admits a Zeno-like interpretation: more frequent stabilizer projections discard more shots, but they also suppress the surviving logical channel's effective noise faster than the post-selection penalty grows. The net result is that the PEC variance factor decreases more rapidly than the inverse acceptance rate increases, so that the product $C=S_1 S_2$ shrinks with decreasing detection interval $\tau$. Our analytically solvable toy models make this trade-off transparent: in the frequent-detection limit, QED+PEC has an exponential coefficient strictly below that of optimal pure PEC, showing that subspace structure together with cheap detection can confer a qualitative advantage which pure PEC cannot recover by changing the cancellation schedule alone.

\subsection{Limitations}
\label{sec:discussion_limitations}

Several caveats should be kept in mind when interpreting our results. First, the main theory and the large-scale numerical demonstration assume ideal stabilizer measurements. The readout-only test shows that classical syndrome flips mainly increase the sampling cost; it is not a faithful model of Iceberg syndrome extraction, because the Iceberg stabilizers are global. Second, when a circuit-level GHZ-assisted extraction circuit is included, the Iceberg benchmark can become post-selection dominated and can lose its sampling advantage over pure PEC. In the fixed-$n$ sweep, the optimal $T$ is finite, showing that smaller $T$ is not generally better once extraction is noisy. This is a limitation of high-weight global checks with noisy CNOT-based extraction, not a proof that QED+PEC is unsuitable for all stabilizer codes. Third, the simulations focus on the $[[n,n{-}2,2]]$ Iceberg code and a single logical GHZ-preparation benchmark; other codes, especially low-weight-check or qLDPC-like codes, may have different extraction-cost trade-offs. Fourth, the numerical demonstration implements only the first-order inverse channel; higher-order implementations and full accepted-instrument modeling of noisy extraction remain to be assessed systematically. Finally, all sampling-cost comparisons exclude the one-time cost of learning the sparse noise model and compiling the PEC tables; this is appropriate when the characterization can be amortized over many circuits or observables, but should be included in a full hardware-resource accounting.

\subsection{Comparison with other QEM methods}
\label{sec:discussion_qem_comparison}

Although we do not claim that QED+PEC is universally superior to every quantum error mitigation method in both sampling complexity and accuracy, nor do we claim that this algorithm saturates the fundamental limit of QEM~\cite{tsubouchi2023universal}, it is instructive to compare its resource trade-off with leading QEM and early fault-tolerant approaches. It should not be interpreted as a hardware-independent total-resource ranking, because QED+PEC additionally requires encoding and syndrome extraction, TEM requires classical tensor-network contraction, and ZNE requires experiments at several amplified noise strengths~\cite{QEM_review,filippov2023scalable,filippov2024scalability}.

For one noisy gate, consider
\[
\mathcal N(\rho)
=
(1-w)\rho+\sum_i p_i\sigma_i\rho\sigma_i^\dagger,
\qquad
w=\sum_i p_i.
\]

We define $\Gamma$ as the sampling cost increase by different QEM algorithms compared with noiseless cases.

\paragraph{PEC.}
To first order,
\[
\mathcal N^{-1}
\simeq
(1+w)\mathcal I-\sum_i p_i\mathcal S_i,
\]
with quasiprobability norm
\[
\gamma=1+2w.
\]
Therefore,
\begin{equation}
\Gamma_{\mathrm{PEC}}
\simeq
\prod_i(1+2w_i)^2
\simeq
e^{4\sum_i w_i}.
\label{eq:qem_comparison_pec_cost}
\end{equation}

\paragraph{QED + PEC.}
Separate the error weight into detectable and non-detectable parts,
\[
W_{\mathrm{detect}},
\qquad
W_{\mathrm{non\text{-}detect}}.
\]
Postselection succeeds with
\[
P_{\mathrm{succ}}
\simeq
1-W_{\mathrm{detect}},
\]
giving
\[
\Gamma_{\mathrm{post}}
\simeq
e^{W_{\mathrm{detect}}}.
\]
PEC on the remaining errors gives
\[
\Gamma_{\mathrm{PEC}}
\simeq
e^{4W_{\mathrm{non\text{-}detect}}}.
\]
Thus,
\begin{equation}
\Gamma_{\mathrm{QED+PEC}}
\simeq
\exp\left[
\sum_i
\left(
W_{\mathrm{detect},i}
+
4W_{\mathrm{non\text{-}detect},i}
\right)
\right].
\label{eq:qem_comparison_qed_pec_cost}
\end{equation}

\paragraph{TEM.}
TEM estimates
\[
\Tr\!\left[
\mathcal N(\rho)\,
\mathcal N^{\dagger,-1}(A)
\right].
\]
Only errors that do not commute with $A$ suppress the signal. For one location,
\[
\text{signal amplification}
\sim
1+2W_{\mathrm{non\text{-}commute}}.
\]
Since sampling complexity scales with the variance,
\[
\Gamma_{\mathrm{TEM}}
\sim
\left(
1+2W_{\mathrm{non\text{-}commute}}
\right)^2.
\]
Therefore, for the whole circuit,
\begin{equation}
\Gamma_{\mathrm{TEM}}
\simeq
\exp\left(
4\sum_i W_{\mathrm{non\text{-}commute},i}
\right).
\label{eq:qem_comparison_tem_cost}
\end{equation}

When the propagated observable samples approximately equal total weights of commuting and anticommuting fault components, the sample increase recovers the typical square root reduction discussed in~\cite{filippov2023scalable,filippov2024scalability}. These studies also established asymptotic sampling optimality of compression-free TEM under a specific noise model.

\paragraph{ZNE.}

For optimized log-ZNE, we choose the noise amplification factor that avoids any additional exponential penalty. Therefore,
\begin{equation}
\Gamma_{\mathrm{ZNE}}
=
\operatorname{poly}\!\left(\sum_i W_{\mathrm{non\text{-}commute},i}\right)
\exp\!\left(
4\sum_i W_{\mathrm{non\text{-}commute},i}
\right).
\label{eq:qem_comparison_zne_cost}
\end{equation}

These methods also introduce controllable systematic errors through different approximations. In QED+PEC, the physical post-selected channel is not truncated; rather, the inverse channel compiled for PEC is truncated at perturbative fault order $K$. Theorem~\ref{thm:k_order_error_bound} gives the end-to-end channel-error bound $\prod_{m=1}^{M}(1+\epsilon_{m,K})-1$. In TEM, the exact inverse map $\mathcal{M}$ is represented as a tensor network and compressed to a finite bond dimension $\chi$. If $\mathcal{M}_{\chi}$ denotes the compressed map, the resulting observable-dependent compression bias is
\begin{equation}
    \delta_{\mathrm{TEM}}^{(\chi)}
    =
    \left|
        \Tr\!\left[
            O\,
            \bigl(
                \mathcal{M}_{\chi}
                -
                \mathcal{M}
            \bigr)
            (\rho_{\mathrm{noisy}})
        \right]
    \right|,
    \label{eq:qem_comparison_tem_bias}
\end{equation}
which is controlled by the tensor-network map approximation error. Its convergence can be assessed by increasing $\chi$, at greater classical memory and contraction cost~\cite{filippov2023scalable,filippov2024scalability}. The TEM bond dimension is therefore a numerical-rank cutoff rather than a perturbative fault-order cutoff. Thus, $K$ in QED+PEC and $\chi$ in TEM both control a bias--cost trade-off, but they truncate different mathematical objects: fault-order contributions to an inverse channel and singular-value sectors of a tensor network, respectively. Noise-model uncertainty and temporal drift constitute a further error source shared by the noise-aware implementations discussed here. Comparisons at a common total precision must account for both terms in $\mathrm{MSE}(\hat o)=[\mathrm{Bias}(\hat o)]^2+\mathrm{Var}(\hat o)$; a smaller sample cost factor alone does not ensure fewer samples at that precision.

\subsection{Outlook}
\label{sec:discussion_outlook}

Several directions naturally extend the present work.

\paragraph{Non-Clifford logical circuits.}
The current construction relies on the fact that Pauli faults propagate as Pauli faults under Clifford logical gates, so that stabilizer commutation can be checked efficiently with the tableau formalism. Extending QED+PEC to native non-Clifford logical blocks will likely require combining the perturbative inverse with Pauli propagation approximations, Clifford-frame expansions, or local tomography of non-Clifford logical blocks.

Universal computation can also be implemented using Clifford operations together with injected magic states, where the injection gadgets involve mid circuit measurements and measurement conditioned Clifford corrections~\cite{zhou2000methodology,bravyi2005universal,litinski2019game}. Since the gates within these injection gadgets are all Clifford, Pauli faults remain Pauli operators when propagated through each individual gate. At first sight, this suggests that the tableau formalism might also simplify the calculations for logical gates implemented by magic state injection. However, a straightforward extension of our method does not follow, because a Pauli fault propagated through a complete logical $T$ gadget can become a non-Pauli output error. To illustrate this, consider a simple injection circuit in which the data qubit controls a CNOT targeting an ancilla prepared in $T|+\rangle$, followed by a $Z$ measurement of the ancilla and an $S$ correction on the data qubit if the measurement outcome is $1$. An $X$ error on the data qubit before the CNOT propagates to an $X$ on both qubits. Although these propagated errors are Pauli operators, the ancilla error flips the measurement outcome and hence changes the conditional correction. Relative to the corresponding error-free branch, the $S$ correction is therefore omitted when it would have been applied, or applied when it would have been omitted. Compensating for this change requires an additional $S$ or $S^\dagger$ operation, so tracking the residual $X$ error alone is insufficient. Accounting for both the propagated error and the altered correction gives $SX$ or $XS^\dagger$, depending on the measurement branch. Up to a global phase, both equal $TXT^\dagger=(X+Y)/\sqrt{2}$, which is not a Pauli operator. Therefore, our algorithm does not directly extend to general injection circuits. An extension requires efficient treatment of measurement conditioned non-Pauli residual channels and their PEC inverses.

In contrast, deterministic logical Clifford gadgets using ideal Pauli measurements and outcome-conditioned Pauli corrections preserve Pauli-error closure, since an incorrectly selected Pauli correction still produces a Pauli residual.
Such gadgets can be incorporated into QED+PEC with feed-forward either within or between blocks, provided their ideal logical action preserves the code space and the relevant measurement records and induced error correlations are retained in the noise analysis.

\paragraph{Noisy syndrome extraction.}
A natural next step is to treat noisy syndrome extraction as an accepted measurement instrument inside the QED+PEC block. At fixed perturbative order, this does not change the polynomial preprocessing principle. The main challenge is physical: extraction faults reduce the observed acceptance probability and add accepted residual logical noise. The Iceberg stress test in Sec.~\ref{sec:syndrome_noise_checks} shows that global checks can make this cost prohibitive under CNOT-based extraction.

\paragraph{Soft information.}
Analog reliability information from noisy syndrome measurements or from an inner continuous-variable code in a concatenated CV--DV architecture can improve outer-code decoding and joint data--syndrome correction~\cite{raveendran2022finite,borah2025fault,raveendran2022soft,pattison2021improved}. In QED+PEC, such information could refine both the acceptance rule and the characterization of the accepted residual channel. A modified acceptance rule may reduce false rejections but can also increase the residual noise requiring cancellation. Its benefit should therefore be assessed through the combined post-selection and PEC cost, together with the resulting bias; a quantitative analysis is left to future work.

\paragraph{Other codes.}
It is important to test the protocol on stabilizer codes with cheaper syndrome extraction. The negative noisy-extraction result for Iceberg is tied to its weight-$n$ global stabilizers. Low-weight-check codes, subsystem codes, and qLDPC-like codes may reduce the extraction-induced acceptance penalty and are natural candidates for future QED+PEC benchmarks.

\paragraph{Code distance and optimal perturbative order.}
The relationship between the code distance $d$, the optimal PEC order $K$, the residual logical mass, the acceptance probability, and the total sampling overhead deserves a systematic study. In particular, the heuristic $K\le d-1$ that emerges naturally for distance-two codes should be tested under richer circuit-level noise, where it may no longer be optimal.

\paragraph{Higher-order implementations.}
Implementing $K>1$ perturbative inverses and benchmarking the predicted $O(W^{K+1})$ bias reduction against the actual quasiprobability overhead will quantify how the bias-versus-cost frontier moves as one climbs the perturbative ladder.

\paragraph{Hardware-oriented validation.}
Finally, QED+PEC should be validated using experimentally learned sparse Pauli--Lindblad noise models, either on real devices or in hardware-faithful simulations, in order to delineate the noise and circuit-volume window in which the protocol is genuinely useful.

\subsection{Closing remarks}
\label{sec:discussion_closing}

Taken together, QED+PEC offers a feedback-free framework for combining stabilizer post-selection with probabilistic cancellation of the residual logical channel. The ideal-measurement Iceberg benchmark shows the possible size of this advantage when detection is cheap. The noisy-extraction stress test identifies the main practical constraint: for global-check codes, syndrome extraction itself can dominate the sampling cost. Finding code families and hardware settings with sufficiently cheap detection is therefore the central next step.

\section*{Note added}
\label{sec:note_added}

In the course of completing this manuscript we became aware of several closely related works on combining subspace post-selection, error detection, and error mitigation with PEC~\cite{papic2026snt,32,33}. Our analysis is complementary in scope: we focus on the post-selected sparse-fault expansion, derive an explicit perturbative inverse channel at arbitrary fixed order, prove explicit sampling-cost and diamond-norm bounds, and analyze the detection-frequency trade-off and the Zeno-like cheap-detection mechanism through large-scale ideal-measurement Iceberg-code simulations, syndrome-noise stress tests, and analytically solvable toy models.

\begin{acknowledgements}
This work is supported by Shanghai Municipal Science and Technology (Grant No. 25LZ2600200), Beijing Natural Science Foundation (Grants No. Z220002), National Natural Science Foundation of China (Grants No. 92365111), and the Quantum Science and Technology-National Science and Technology Major Project (Grant No. 2021ZD0302400).
\end{acknowledgements}

\section*{Data availability}

The data and source code supporting the findings of this study are openly available in the GitHub repository at \url{https://github.com/yuanyithu/zeno_qedc_pec} and are permanently archived in Zenodo. The specific version associated with this article is available at \href{https://doi.org/10.5281/zenodo.22283754}{doi:10.5281/zenodo.22283754} (version DOI); the record covering all versions is available at \href{https://doi.org/10.5281/zenodo.22283753}{doi:10.5281/zenodo.22283753} (concept DOI).

\bibliography{refs}

\onecolumngrid
\newpage
\appendix
\appendixtableofcontents
\newpage

\section{Additional details for the algorithmic \texorpdfstring{order-$K$}{order-K} QED+PEC construction}
\label{app:k_order_algorithmic_bound}

\subsection{One-block notation}

We first work within one block and suppress the block label $m$.  Let
$\mathcal K$ be the set of independent physical Pauli fault locations.  Each
location $i\in\mathcal K$ occurs with probability $w_i$, and a subset
$I\subseteq\mathcal K$ has exact Bernoulli probability
\begin{equation}
    \pi_I
    =
    w_I\prod_{j\in\mathcal K\setminus I}(1-w_j),
    \qquad
    w_I:=\prod_{i\in I}w_i .
\end{equation}
After propagation through the Clifford block, the fault branch acts as a
Pauli channel $\mathcal P_I(\cdot)=\tilde P_I(\cdot)\tilde P_I^\dagger$.
The branch is accepted by QED iff $\tilde P_I$ commutes with every stabilizer
generator.  We write
\begin{equation}
    \mathcal K_{\mathrm{pass}}
    :=
    \{I\subseteq\mathcal K:[\tilde P_I,S_a]=0\ \text{for all }a\}.
\end{equation}
For an input state already in the code space, an accepted Pauli branch remains
inside the code space, while a rejected branch is removed by the syndrome
projection.  Therefore the exact unnormalized accepted channel and exact
success probability are
\begin{equation}
    \bar N
    :=
    \sum_{I\in\mathcal K_{\mathrm{pass}}}\pi_I\mathcal P_I,
    \qquad
    p_{\mathrm{success}}
    :=
    \sum_{I\in\mathcal K_{\mathrm{pass}}}\pi_I.
\end{equation}
The exact normalized hardware residual channel is
\begin{equation}
    N_{\mathrm{reduce}}^{\mathrm{exact}}
    :=
    \frac{\bar N}{p_{\mathrm{success}}}.
    \label{eq:appendix_exact_reduced_algorithmic}
\end{equation}
This channel contains all accepted fault orders.

For $|I|\le K$, the algorithmic branch coefficient is
\begin{equation}
    a_I^{(K)}
    :=
    \operatorname{Trunc}_{\le K}\left[
        w_I\prod_{j\in\mathcal K\setminus I}(1-w_j)
    \right]
    =
    w_I
    \sum_{\ell=0}^{K-|I|}(-1)^\ell
    \sum_{\substack{J\subseteq\mathcal K\setminus I\\ |J|=\ell}}w_J .
    \label{eq:appendix_algorithmic_aI}
\end{equation}
The algorithm forms
\begin{equation}
    \bar N_K
    :=
    \sum_{\substack{I\in\mathcal K_{\mathrm{pass}}\\ |I|\le K}}
    a_I^{(K)}\mathcal P_I,
    \qquad
    P_{K,\mathrm{success}}
    :=
    \sum_{\substack{I\in\mathcal K_{\mathrm{pass}}\\ |I|\le K}}
    a_I^{(K)}.
\end{equation}
The Taylor-truncated reduced channel used for compiling the PEC table is
\begin{equation}
    \hat{N}_{\mathrm{reduce}}
    :=
    \operatorname{Trunc}_{\le K}\left[
        \frac{\bar N_K}{P_{K,\mathrm{success}}}
    \right]
    =
    \operatorname{id}+R_K.
    \label{eq:appendix_AK_definition}
\end{equation}
Finally the implemented quasiprobability channel is
\begin{equation}
    \widehat N_K^{-1}
    :=
    \operatorname{Collect}_P
    \operatorname{Trunc}_{\le K}\left[
        \sum_{r=0}^{K}(-1)^r R_K^{\circ r}
    \right]
    =
    \sum_P c_P^{(K)}\mathcal P,
    \qquad
    \gamma_K:=\sum_P|c_P^{(K)}|.
    \label{eq:appendix_algorithmic_inverse}
\end{equation}
Since each $\mathcal P$ is a unitary Pauli channel,
\begin{equation}
    \|\widehat N_K^{-1}\|_\diamond
    \le
    \sum_P|c_P^{(K)}|=\gamma_K.
    \label{eq:appendix_gamma_norm_bound}
\end{equation}

The one-block diamond-norm inequality and the multi-block composition bound
are established in the proof sketch of Theorem~\ref{thm:k_order_error_bound}
and are not repeated in this appendix.

\subsection{The algorithmic reduced channel is the Taylor truncation of the exact reduced channel}

\begin{lemma}[Identification of $\hat{N}_{\mathrm{reduce}}$]
\label{lem:AK_is_exact_taylor_truncation}
As a formal power series in the fault probabilities $\{w_i\}$,
\begin{equation}
    \hat{N}_{\mathrm{reduce}}
    =
    \operatorname{Trunc}_{\le K}
    \left[N_{\mathrm{reduce}}^{\mathrm{exact}}\right].
    \label{eq:AK_equals_exact_reduced_truncation}
\end{equation}
\end{lemma}

\begin{proof}
For $|I|\le K$, Eq.~\eqref{eq:appendix_algorithmic_aI} is precisely
$\operatorname{Trunc}_{\le K}[\pi_I]$.  If $|I|>K$, the exact branch
probability $\pi_I$ has total degree at least $|I|>K$, so
$\operatorname{Trunc}_{\le K}[\pi_I]=0$.  Hence
\begin{equation}
    \bar N_K=\operatorname{Trunc}_{\le K}[\bar N],
    \qquad
    P_{K,\mathrm{success}}
    =
    \operatorname{Trunc}_{\le K}[p_{\mathrm{success}}].
\end{equation}
Moreover $p_{\mathrm{success}}(0)=1$, because the no-fault branch is accepted.
For any scalar formal series $p=1+O(w)$, the coefficients through degree $K$ of
$1/p$ depend only on the coefficients through degree $K$ of $p$.  Therefore
\begin{align}
    \operatorname{Trunc}_{\le K}
    \left[\frac{\bar N}{p_{\mathrm{success}}}\right]
    &=
    \operatorname{Trunc}_{\le K}
    \left[
        \frac{\operatorname{Trunc}_{\le K}\bar N}
             {\operatorname{Trunc}_{\le K}p_{\mathrm{success}}}
    \right] \\
    &=
    \operatorname{Trunc}_{\le K}
    \left[
        \frac{\bar N_K}{P_{K,\mathrm{success}}}
    \right]
    =\hat{N}_{\mathrm{reduce}} .
\end{align}
This proves the claim.
\end{proof}

Lemma~\ref{lem:AK_is_exact_taylor_truncation} is where the proof differs from
the earlier exact-$\pi_I$ retained-channel argument.  The channel
normalized by exact branch probabilities over $|I|\le K$ is not used here.

\subsection{The implemented inverse is a formal inverse through \texorpdfstring{order $K$}{order K}}

\begin{lemma}[Formal inverse property]
\label{lem:formal_inverse_property}
Let $\hat{N}_{\mathrm{reduce}}=\operatorname{id}+R_K$ be defined by
Eq.~\eqref{eq:appendix_AK_definition}.  Then
\begin{equation}
    \operatorname{Trunc}_{\le K}
    \left[
        \widehat N_K^{-1}\circ\hat{N}_{\mathrm{reduce}}
    \right]
    =
    \operatorname{id}.
    \label{eq:formal_inverse_to_order_K}
\end{equation}
Consequently,
$\widehat N_K^{-1}\circ\hat{N}_{\mathrm{reduce}}-\operatorname{id}$ has no monomial of
total degree $0,1,\ldots,K$.
\end{lemma}

\begin{proof}
The channel $R_K$ has no degree-zero term, because $\hat{N}_{\mathrm{reduce}}$ is a
trace-preserving Taylor truncation of a channel that reduces to
$\operatorname{id}$ at zero noise.  In the ring of formal channel-valued power series,
\begin{equation}
    (\operatorname{id}+R_K)^{-1}
    =
    \sum_{r=0}^{\infty}(-1)^rR_K^{\circ r}.
\end{equation}
Terms with $r>K$ have total degree at least $K+1$.  Therefore the degree-$K$
truncation of this formal inverse is exactly the expression implemented in
Eq.~\eqref{eq:appendix_algorithmic_inverse}.  Multiplying by
$\hat{N}_{\mathrm{reduce}}=\operatorname{id}+R_K$ and truncating through degree $K$ gives
Eq.~\eqref{eq:formal_inverse_to_order_K}.
\end{proof}

For numerical certification, one does not need to manipulate the diamond norm
symbolically.  After collecting Pauli channels, write
\begin{equation}
    \widehat N_K^{-1}\circ\hat{N}_{\mathrm{reduce}}-\operatorname{id}
    =
    \sum_P b_P^{(K)}\mathcal P .
\end{equation}
Then
\begin{equation}
    \zeta_K
    :=
    \left\|
        \widehat N_K^{-1}\circ\hat{N}_{\mathrm{reduce}}-\operatorname{id}
    \right\|_\diamond
    \le
    \sum_P |b_P^{(K)}|,
    \label{eq:zeta_l1_certificate_appendix}
\end{equation}
and the right-hand side is obtained by the same Pauli-channel collection used
in the PEC preprocessing.  Lemma~\ref{lem:formal_inverse_property} implies
that this coefficient sum begins at total degree $K+1$.

\subsection{Perturbative scaling}
\label{app:perturbative_scaling}

This subsection proves Corollary~\ref{cor:perturbative_scaling}.  Let
$W:=\sum_{i\in\mathcal K}w_i$ and assume that the success probability
$p_{\mathrm{success}}$ stays bounded away from zero as $W\to0$.  We
estimate the three certificates separately.

\paragraph{Estimate for $\eta_K$.}
Scale all fault probabilities by a single parameter $t$, defining
$N_{\mathrm{reduce}}^{\mathrm{exact}}(t):=N_{\mathrm{reduce}}^{\mathrm{exact}}(\{tw_i\})$.
By Lemma~\ref{lem:AK_is_exact_taylor_truncation},
\begin{equation}
    \hat{N}_{\mathrm{reduce}}
    =
    \sum_{q=0}^{K}\frac{1}{q!}
    \left.\frac{d^q}{dt^q}N_{\mathrm{reduce}}^{\mathrm{exact}}(t)
    \right|_{t=0}.
\end{equation}
The integral Taylor remainder evaluated at $t=1$ gives
\begin{equation}
    \eta_K
    \le
    \frac{1}{K!}
    \int_0^1(1-t)^K
    \left\|
        \frac{d^{K+1}}{dt^{K+1}}
        N_{\mathrm{reduce}}^{\mathrm{exact}}(t)
    \right\|_\diamond dt.
    \label{eq:eta_taylor_remainder_appendix}
\end{equation}
The integrand is a sum of branch contributions whose $w$-degree is at
least $K+1$, so the right-hand side is $O(W^{K+1})$ whenever
$p_{\mathrm{success}}$ stays bounded away from zero on the segment
$\{tw_i:0\le t\le1\}$.

\paragraph{Estimate for $\zeta_K$.}
By Lemma~\ref{lem:formal_inverse_property}, the operator
$\widehat N_K^{-1}\circ\hat{N}_{\mathrm{reduce}}-\operatorname{id}$ has
no monomials of total $w$-degree $\le K$ when expanded in the variables
$\{w_i\}$.  Collecting Pauli channels and using
Eq.~\eqref{eq:zeta_l1_certificate_appendix},
\begin{equation}
    \zeta_K
    \le
    \sum_P|b_P^{(K)}|=O(W^{K+1}).
\end{equation}

\paragraph{Estimate for $\gamma_K$.}
Since $R_K=\hat{N}_{\mathrm{reduce}}-\operatorname{id}$ has every Pauli
coefficient of order at least $W$, the Neumann expansion
\begin{equation}
    \widehat N_K^{-1}
    =
    \operatorname{Trunc}_{\le K}\!\left[\operatorname{id}+\sum_{r=1}^{K}(-1)^rR_K^{\circ r}\right]
\end{equation}
gives an identity coefficient of $1+O(W)$, while the absolute sum of
the nonidentity coefficients is $O(W)$.  Therefore
\begin{equation}
    \gamma_K
    =
    \sum_P|c_P^{(K)}|
    =
    1+O(W).
\end{equation}

\paragraph{Combined estimate.}
The one-block inequality in Theorem~\ref{thm:k_order_error_bound} gives
\begin{equation}
    \epsilon_K
    =
    \zeta_K+\gamma_K\eta_K .
\end{equation}
Using the estimates above,
\[
    \eta_K=O(W^{K+1}),
    \qquad
    \zeta_K=O(W^{K+1}),
    \qquad
    \gamma_K=1+O(W),
\]
we obtain
\begin{equation}
    \epsilon_K=O(W^{K+1}).
\end{equation}
Restoring block labels and using the multi-block bound in
Theorem~\ref{thm:k_order_error_bound}, whenever
$\sum_m\epsilon_{m,K}\ll 1$,
\begin{equation}
    \left\|
        \mathcal C_{\mathrm{QED+PEC}}^{(K)}
        -\mathcal C_{\mathrm{ideal}}
    \right\|_\diamond
    =
    O\left(
        \sum_{m=1}^M W_m^{K+1}
    \right),
\end{equation}
which proves Corollary~\ref{cor:perturbative_scaling}.

\subsection{Implementation remarks}
\label{app:implementation_remarks}

The fixed-order construction is implemented using sparse branch tables rather
than a full Pauli-channel table over the $4^n$ Pauli group.

\begin{enumerate}
    \item \textbf{Sparse branch representation.}
    Each retained branch is stored by its fault-label subset
    $I\subseteq\mathcal K$ with $|I|\le K$, its polynomial coefficient
    $a_I^{(K)}$, and its propagated symplectic Pauli string $\tilde P_I$.
    The algorithm never enumerates fault subsets of size larger than $K$.

    \item \textbf{Stabilizer acceptance check.}
    A retained branch is accepted exactly when $\tilde P_I$ commutes with all
    stabilizer generators.  In the Clifford setting this is checked by
    symplectic inner products, so the acceptance test is independent of the
    Hilbert-space dimension.

    \item \textbf{Normalization and inverse coefficients.}
    After accepted branches are collected, the truncated normalized channel
    $\hat N_{\mathrm{reduce}}=\operatorname{id}+R_K$ is formed by polynomial
    arithmetic through total degree $K$.  The implemented inverse is the
    truncated Neumann series
    \[
        \operatorname{Trunc}_{\le K}
        \left[
            \sum_{r=0}^K(-1)^rR_K^{\circ r}
        \right].
    \]

    \item \textbf{PEC sampling table.}
    Identical propagated Pauli channels are collected to obtain coefficients
    $c_P^{(K)}$.  The table stores
    \[
        \left\{
            \left(P,\frac{|c_P^{(K)}|}{\gamma_K},
            \operatorname{sgn}(c_P^{(K)})\right):
            c_P^{(K)}\ne 0
        \right\},
        \qquad
        \gamma_K=\sum_P|c_P^{(K)}|.
    \]

    \item \textbf{Practical checks.}
    In numerical use, one can verify the table by checking that the collected
    coefficients of
    $\widehat N_K^{-1}\circ \hat N_{\mathrm{reduce}}-\operatorname{id}$
    vanish through total degree $K$, as guaranteed by
    Lemma~\ref{lem:formal_inverse_property}.
\end{enumerate}

\section{Toy-model algebra checks}
\label{app:toy_model}

The main physical derivation of the two toy models is given in
Sec.~\ref{sec:toy_model}.  This appendix records the intermediate algebraic
checks and Taylor expansions used to obtain Eqs.~\eqref{eq:toy_totalA},
\eqref{eq:toy_totalB}, and \eqref{eq:toy_purePEC_bound}.  Keeping these
details here makes the coefficients reproducible without forcing the main text
to carry every expansion step.

\subsection{Model A: QED+PEC}
\label{app:toy:modelA}

\subsubsection{Depolarizing dynamics and single-round survival probability}
\label{app:toy:dynamics}
Consider an initial pure state $\rho_0$ supported entirely on the code subspace,
with projector $P$ satisfying $\mathrm{Tr}(P)=2$ and $P\rho_0 P=\rho_0$.
The system evolves under a global depolarizing Lindbladian
\begin{equation}
\frac{d\rho}{dt}=\mathcal{L}(\rho)
=\gamma\!\left(\frac{\mathrm{Tr}(\rho)}{N}\,I-\rho\right),
\label{eq:app_lindblad}
\end{equation}
whose exact solution over a time interval $\tau$ is
\begin{equation}
\rho(\tau)=e^{-\gamma\tau}\rho_0
+\bigl(1-e^{-\gamma\tau}\bigr)\frac{I}{N}\,.
\label{eq:app_rho_tau}
\end{equation}
Defining the depolarizing probability $p=1-e^{-\gamma\tau}$, this becomes
\begin{equation}
\rho(\tau)=(1-p)\,\rho_0+\frac{p}{N}\,I\,.
\label{eq:app_rho_p}
\end{equation}

At the end of each interval, we apply the code-subspace projector $P$.
Using $P\rho_0 P=\rho_0$ and $\mathrm{Tr}(P)=2$, the unnormalized
post-selected state is
\begin{equation}
\tilde{\rho}_P(\tau)=P\rho(\tau)P
=(1-p)\,\rho_0+\frac{p}{N}\,P\,,
\label{eq:app_rhoP}
\end{equation}
and the single-round survival (post-selection success) probability is
\begin{equation}
p_{\mathrm{succ}}
=\mathrm{Tr}\!\bigl(\tilde{\rho}_P(\tau)\bigr)
=1-p+\frac{2p}{N}
=1-p\!\left(1-\frac{2}{N}\right).
\label{eq:app_psucc}
\end{equation}

\subsubsection{Normalized post-selected state and the effective Pauli channel}
\label{app:toy:pauli_channel}

To construct the PEC inverse, we need the normalized post-selected state
expressed as a Pauli channel acting on $\rho_0$.
Within the two-dimensional code subspace, the projector admits the
decomposition
\begin{equation}
P = \frac{1}{2}\bigl(\rho_0 + X\rho_0 X + Y\rho_0 Y + Z\rho_0 Z\bigr)
\equiv \frac{1}{2}\rho_0 + \mathcal{D}(\rho_0)\,,
\label{eq:app_P_decomp}
\end{equation}
where $\mathcal{D}(\rho_0)=\frac{1}{2}(X\rho_0 X+Y\rho_0 Y+Z\rho_0 Z)$
denotes the completely dephasing (off-diagonal killing) part.
Substituting into Eq.~\eqref{eq:app_rhoP}:
\begin{equation}
\tilde{\rho}_P(\tau)
=\left(1-p+\frac{p}{2N}\right)\rho_0
+\frac{p}{N}\,\mathcal{D}(\rho_0)\,.
\label{eq:app_rhoP_pauli}
\end{equation}
After normalization by $p_{\mathrm{succ}}$, this defines a CPTP Pauli channel
$\mathcal{E}$ with coefficients
\begin{equation}
A=\frac{1-p+\frac{p}{2N}}{p_{\mathrm{succ}}}\,,
\qquad
B=\frac{\frac{p}{2N}}{p_{\mathrm{succ}}}\,,
\label{eq:app_AB}
\end{equation}
where $A$ is the weight of the identity operation and $B$ is the weight of
each nontrivial Pauli ($X$, $Y$, $Z$).

\subsubsection{Single-round PEC quasiprobability overhead}
\label{app:toy:pec}

In the Pauli transfer matrix (PTM) representation, the effective channel
$\mathcal{E}$ has degenerate eigenvalues
$\lambda_x=\lambda_y=\lambda_z\equiv\lambda$, with
\begin{equation}
\lambda = A - B
= \frac{1-p}{p_{\mathrm{succ}}}
= \frac{1-p}{1-p\!\left(1-\frac{2}{N}\right)}\,.
\label{eq:app_lambda}
\end{equation}
The PEC quasiprobability overhead for inverting a single round is
$\gamma_{\mathrm{PEC}}=\tfrac{3}{2}\lambda^{-1}-\tfrac{1}{2}$.
Inverting Eq.~\eqref{eq:app_lambda}:
\begin{equation}
\lambda^{-1}
=\frac{1-p+\frac{2p}{N}}{1-p}
=1+\frac{2}{N}\,\frac{p}{1-p}\,.
\label{eq:app_lambda_inv}
\end{equation}
Define the auxiliary variable
\begin{equation}
x \;:=\; \frac{p}{1-p}
\;=\; e^{\gamma\tau}-1\,,
\label{eq:app_x_def}
\end{equation}
where the last equality follows from $p=1-e^{-\gamma\tau}$.
Substituting into the PEC overhead:
\begin{equation}
\gamma_{\mathrm{PEC}}
= \frac{3}{2}\!\left(1+\frac{2}{N}\,x\right)-\frac{1}{2}
= 1+\frac{3}{N}\bigl(e^{\gamma\tau}-1\bigr).
\label{eq:app_gamma_pec}
\end{equation}

\subsubsection{Total sample complexity after \texorpdfstring{$T/\tau$}{T/tau} rounds}
\label{app:toy:total}

After $T/\tau$ consecutive rounds of noise-then-project, the total
post-selection cost and PEC variance overhead are
\begin{equation}
S_1 = \bigl(p_{\mathrm{succ}}\bigr)^{-T/\tau},
\qquad
S_2 = \bigl(\gamma_{\mathrm{PEC}}^{2}\bigr)^{T/\tau}.
\label{eq:app_S1S2_def}
\end{equation}
The total sample complexity scales as $C_{\mathrm{total}}\propto S_1\,S_2$.
We evaluate $\ln(S_1 S_2)=\ln S_1+\ln S_2$ by expanding each factor to
$\mathcal{O}(\tau^2)$ in the single-round logarithm, since multiplication
by $T/\tau$ promotes the $\mathcal{O}(\tau^2)$ correction to an
$\mathcal{O}(\tau)$ contribution in the total exponent.

\paragraph{Post-selection cost $\ln S_1$.}
Write $c=1-2/N$, so that $p_{\mathrm{succ}}=1-c\,p=1-c(1-e^{-\gamma\tau})$.
Expanding $1-e^{-\gamma\tau}=\gamma\tau-\tfrac{1}{2}\gamma^2\tau^2
+\mathcal{O}(\tau^3)$ and using
$\ln(1-u)\approx -u-\tfrac{1}{2}u^2$:
\begin{equation}
\ln p_{\mathrm{succ}}
\approx -c\!\left(\gamma\tau-\tfrac{1}{2}\gamma^2\tau^2\right)
-\tfrac{1}{2}c^2(\gamma\tau)^2
= -c\,\gamma\tau+\tfrac{1}{2}c(1-c)\,\gamma^2\tau^2.
\label{eq:app_ln_psucc}
\end{equation}
Multiplying by $-T/\tau$ and substituting $c(1-c)=\frac{2}{N}-\frac{4}{N^2}$:
\begin{equation}
\ln S_1
= \gamma T\!\left(1-\frac{2}{N}\right)
+\gamma^2 T\tau\!\left(-\frac{1}{N}+\frac{2}{N^2}\right).
\label{eq:app_lnS1}
\end{equation}

\paragraph{PEC variance overhead $\ln S_2$.}
Using $e^{\gamma\tau}-1=\gamma\tau+\tfrac{1}{2}\gamma^2\tau^2
+\mathcal{O}(\tau^3)$:
\begin{align}
\ln\!\bigl(\gamma_{\mathrm{PEC}}^{2}\bigr)
&= 2\ln\!\left[1+\frac{3}{N}\!\left(\gamma\tau
  +\tfrac{1}{2}\gamma^2\tau^2\right)\right]
\notag\\
&\approx 2\left[\frac{3}{N}\!\left(\gamma\tau
  +\tfrac{1}{2}\gamma^2\tau^2\right)
  -\frac{1}{2}\!\left(\frac{3}{N}\gamma\tau\right)^{\!2}\right]
\notag\\
&= \frac{6\gamma\tau}{N}
  +\gamma^2\tau^2\!\left(\frac{3}{N}-\frac{9}{N^2}\right).
\label{eq:app_ln_gamma2}
\end{align}
Multiplying by $T/\tau$:
\begin{equation}
\ln S_2
= \frac{6\gamma T}{N}
+\gamma^2 T\tau\!\left(\frac{3}{N}-\frac{9}{N^2}\right).
\label{eq:app_lnS2}
\end{equation}

\subsubsection{Combined result}
\label{app:toy:combined}

Adding Eqs.~\eqref{eq:app_lnS1} and \eqref{eq:app_lnS2}:
\begin{equation}
\ln C_{\mathrm{QED{+}PEC}}
= \gamma T\!\left(1+\frac{4}{N}\right)
+\gamma^2 T\tau\!\left(\frac{2}{N}-\frac{7}{N^2}\right).
\label{eq:app_lnCtotal}
\end{equation}
In exponential form,
\begin{equation}
C_{\mathrm{QED{+}PEC}}
\approx
\exp\!\left[\gamma T\!\left(1+\frac{4}{N}\right)\right]
\;\exp\!\left[\gamma^2 T\tau\!\left(\frac{2}{N}-\frac{7}{N^2}\right)\right].
\label{eq:app_Ctotal_exp}
\end{equation}
For $N\ge 4$, the coefficient $2/N-7/N^2$ is strictly positive, so the
$\mathcal{O}(\gamma^2 T\tau)$ correction is a penalty that grows with
$\tau$ and vanishes as $\tau\to 0$.

\subsection{Model B: Pure subspace PEC}
\label{app:toy:modelB}

\subsubsection{Setup}
\label{app:toy:modelB_setup}

To provide a baseline for comparison, we now analyze a system that is
restricted to the $2$-dimensional {physical qubit space},
with no ancilla degrees of freedom and hence no possibility of error
detection.
The only error-mitigation tool available is PEC, applied periodically
at intervals $\tau$.

The depolarizing noise within the $2$-dimensional space is governed by

\begin{equation}
\mathcal{L}(\rho)=\gamma\!\left(\frac{I}{2}-\rho\right),
\label{eq:appB_lindblad}
\end{equation}

whose exact solution over an interval $\tau$ is
\begin{equation}
\rho(\tau)=(1-p)\,\rho_0+p\,\frac{I}{2}\,,
\qquad p=1-e^{-\gamma\tau}\,.
\label{eq:appB_rho}
\end{equation}
Since the system never leaves the physical space, there is no
post-selection step and $S_1=1$ identically.

\subsubsection{Single-round PEC overhead}
\label{app:toy:modelB_pec}

Using the Pauli decomposition
$\frac{I}{2}=\frac{1}{4}(\rho_0+X\rho_0 X+Y\rho_0 Y+Z\rho_0 Z)$,
the noisy state can be written as the Pauli channel
\begin{equation}
\mathcal{E}_B(\rho_0)
=\!\left(1-\tfrac{3}{4}p\right)\rho_0
+\tfrac{p}{4}\bigl(X\rho_0 X+Y\rho_0 Y+Z\rho_0 Z\bigr).
\label{eq:appB_pauli}
\end{equation}
The PTM eigenvalues are $\lambda_0=1$ and
$\lambda_x=\lambda_y=\lambda_z=1-p=e^{-\gamma\tau}$, giving the
single-round PEC overhead
\begin{equation}
\gamma_{\mathrm{PEC}}^{(B)}
=\tfrac{3}{2}\,e^{\gamma\tau}-\tfrac{1}{2}\,.
\label{eq:appB_gamma}
\end{equation}

\subsubsection{Total PEC complexity with periodic cancellation}
\label{app:toy:modelB_total}

Applying PEC every $\tau$ over a total time $T$ gives
$C_{\mathrm{pure\;PEC}}=(\gamma_{\mathrm{PEC}}^{(B)\,2})^{T/\tau}$.
Expanding the single-round logarithm to $\mathcal{O}(\tau^2)$:
\begin{align}
\ln\!\bigl(\gamma_{\mathrm{PEC}}^{(B)\,2}\bigr)
&= 2\ln\!\left(\tfrac{3}{2}e^{\gamma\tau}-\tfrac{1}{2}\right)
\notag\\
&= 2\ln\!\left(1+\tfrac{3}{2}\gamma\tau
  +\tfrac{3}{4}\gamma^2\tau^2+\cdots\right)
\notag\\
&\approx 2\left[\tfrac{3}{2}\gamma\tau
  +\tfrac{3}{4}\gamma^2\tau^2
  -\tfrac{1}{2}\!\left(\tfrac{3}{2}\gamma\tau\right)^{\!2}\right]
\notag\\
&= 3\gamma\tau - \tfrac{3}{4}\gamma^2\tau^2.
\label{eq:appB_ln_gamma2}
\end{align}
Multiplying by $T/\tau$:
\begin{equation}
\ln C_{\mathrm{pure\;PEC}}^{(\tau)}
= 3\gamma T - \tfrac{3}{4}\gamma^2 T\tau\,.
\label{eq:appB_total}
\end{equation}
Note that the $\mathcal{O}(\gamma^2 T\tau)$ correction is
\emph{negative}: smaller $\tau$ (more frequent PEC) \emph{increases}
the total cost.
This is the opposite sign compared with Model~A
[Eq.~\eqref{eq:app_lnCtotal}], and reflects the well-known fact that
splitting a noise channel into many small pieces and inverting each one
separately is suboptimal for pure PEC~\cite{31}.

\subsubsection{Optimal single-shot PEC (\texorpdfstring{$\tau=T$}{tau=T})}
\label{app:toy:modelB_single}

The optimal pure-PEC strategy is to accumulate all noise over the full
duration $T$ and apply PEC only once at the end.
Setting $\tau=T$ (a single round) gives
\begin{equation}
\gamma_{\mathrm{PEC}}^{(B)}\big|_{\tau=T}
=\tfrac{3}{2}\,e^{\gamma T}-\tfrac{1}{2}\,,
\label{eq:appB_gamma_single}
\end{equation}
and the total sample complexity is
\begin{equation}
C_{\mathrm{pure\;PEC}}^{(\tau=T)}
=\left(\tfrac{3}{2}\,e^{\gamma T}-\tfrac{1}{2}\right)^{2}
\;\approx\; \tfrac{9}{4}\,e^{2\gamma T}
\quad\text{for }\gamma T\gg 1\,.
\label{eq:appB_single_cost}
\end{equation}
This serves as a lower bound on the pure-PEC sample complexity.


\section{Quantum Fisher Information Analysis}
\label{sec:qfi}

In this section, we address a fundamental question: under local noise, can subspace encoding reduce the sampling cost of expectation-value estimation at the information-theoretic level? Estimating $\mathrm{Tr}(\rho X)$ from repeated measurements on a noisy state is inherently a parameter estimation problem, and the quantum Cram\'{e}r--Rao bound (QCRB) sets the ultimate lower limit on the required number of samples~\cite{braunstein1994,optimal_measurement_2010}. In this context, information has a direct operational meaning: the fewer samples needed to reach a target precision, the more information the measurements carry about the parameter.

Consider an $n$-qubit system whose $2^n$-dimensional Hilbert space $\mathcal{H}$ contains a $2^k$-dimensional code subspace $\mathcal{C}$ defined by the projector $\Pi$. The initial state $\rho = \Pi\rho\Pi$ lies within $\mathcal{C}$, and we wish to estimate $\mathrm{Tr}(\rho X)$ for an observable $X$ supported on $\mathcal{C}$, after the state passes through a local noise channel $\mathcal{N}$. As shown in Ref.~\cite{optimal_measurement_2010}, the QCRB implies that the minimum number of samples required to estimate $\mathrm{Tr}(\rho X)$ to mean-squared error $\varepsilon^2$ satisfies
\begin{equation}
N_{\min} \;\geq\; \frac{\mathrm{Tr}\bigl(X\,\mathcal{M}^{-1}(X)\bigr)}{\varepsilon^2}\,,
\label{eq:qcrb_sampling}
\end{equation}
where $\mathcal{M}$ is a superoperator determined jointly by the noise channel and the subspace structure,
\begin{equation}
\mathcal{M}(\cdot) = \Pi \circ \mathcal{N}^\dagger \circ \mathcal{R}_{\mathcal{N}(\rho)}^{-1} \circ \mathcal{N} \circ \Pi(\cdot)\,,
\label{eq:superop_M}
\end{equation}
with $\mathcal{R}_{\mathcal{N}(\rho)}(\cdot) = \frac{1}{2}\bigl\{\mathcal{N}(\rho),\;\cdot\;\bigr\}$ the symmetric logarithmic derivative (SLD) superoperator of the noisy state. In the absence of noise ($\mathcal{N} = \mathrm{id}$), $\mathcal{M}$ reduces to the identity within $\mathcal{C}$, and Eq.~\eqref{eq:qcrb_sampling} recovers the standard shot-noise limit $N_{\min} \geq \mathrm{Tr}(X^2)/\varepsilon^2$.

The comparison between the noisy and noiseless bounds motivates the \emph{information contraction rate},
\begin{equation}
r = \frac{\mathrm{Tr}(X^2)}{\mathrm{Tr}\bigl(X\,\mathcal{M}^{-1}(X)\bigr)}\,.
\label{eq:contraction_rate}
\end{equation}
Since noise can only degrade information, $r \le 1$: the noise compresses the extractable information, inflating the sampling cost by a factor of $1/r$. Equation~\eqref{eq:contraction_rate} provides clear physical intuition about the role of the subspace. Without encoding ($\Pi = I$), the eigenvalues of $\mathcal{M}$ decay exponentially with circuit depth under local noise, driving $r \to 0$ and causing the exponential sampling explosion of pure PEC. By contrast, a well-chosen nontrivial $\Pi$ reshapes the eigenspectrum of $\mathcal{M}$, turning the complement of $\mathcal{C}$ into an absorbing sink for leaked population and dramatically slowing the decay of $r$. However, the eigenstructure of $\mathcal{M}$ depends nonlinearly on $\Pi$, precluding closed-form optimization and motivating a numerical approach.

The exact computation of $\mathcal{M}$ requires full diagonalization of the $2^n \times 2^n$ density matrix, limiting us to small systems. Within this window, we compare three subspace choices: (i)~the \emph{best-found} subspace from numerical optimization on the Stiefel manifold; (ii)~the $[[n,1]]$ repetition code subspace; and (iii)~the trivial baseline that keeps the original logical subspace unchanged.

\begin{figure}[t]
    \centering
    \includegraphics[width=\columnwidth]{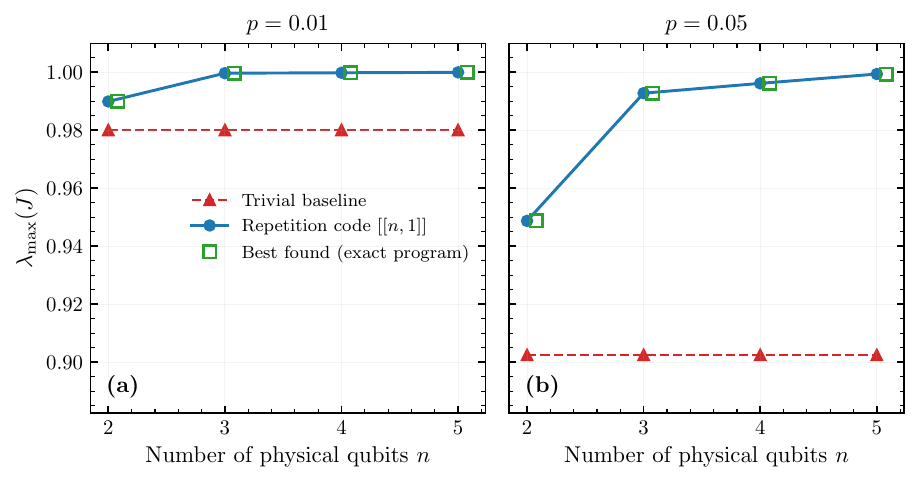}
    \caption{Comparison of the maximal information contraction rate $\lambda_{\max}(J)$ for the trivial baseline, the $[[n,1]]$ repetition code, and the numerically best-found subspace, at two local depolarizing noise strengths. Across all tested system sizes, the repetition code is nearly indistinguishable from the best-found subspace, while the trivial baseline is clearly worse.}
    \label{fig:qfi_comparison}
\end{figure}

Numerically, Fig.~\ref{fig:qfi_comparison} shows that for the maximal contraction rate---that is, along the best-protected logical direction---the repetition code is already nearly indistinguishable from the best-found subspace, while the trivial baseline is clearly worse. This indicates that simple stabilizer structure can already preserve the relevant information remarkably well under local noise, which motivates the protocol-oriented viewpoint adopted below: rather than pursuing an abstract global optimum, we focus on local QEDCs that can turn this protection into a hardware-realizable reduction of sampling overhead.

\end{document}